\documentclass[12pt, journal, onecolumn]{IEEEtran}
\usepackage{epsfig,rotating,setspace,latexsym,amsmath,epsf,amssymb,bm}
\usepackage{graphicx,color}
\usepackage{booktabs}
\usepackage{amsmath,graphicx}
\usepackage{rotating,latexsym,amsmath,amssymb,bm}
\usepackage{cite}
\usepackage{algpseudocode}
\usepackage{algorithmicx}
\usepackage{algorithm}
\usepackage{amsthm}
\usepackage{subfigure}
\usepackage{graphicx}
\usepackage{subfigure}
\usepackage{algpseudocode}
\usepackage{algorithm}
\usepackage{booktabs}
\usepackage{wrapfig}
\usepackage{dsfont}
\usepackage{wasysym}
\usepackage{hyperref}

\newcommand{\myeq}[1]{\mathrel{\overset{\makebox[0pt]{\mbox{\normalfont\tiny\sffamily #1}}}{=}}}
\newcommand{\myleq}[1]{\mathrel{\overset{\makebox[0pt]{\mbox{\normalfont\tiny\sffamily #1}}}{\leq}}}
\newcommand{\mygeq}[1]{\mathrel{\overset{\makebox[0pt]{\mbox{\normalfont\tiny\sffamily #1}}}{\geq}}}
\newcommand{\n}{\noindent}

\theoremstyle{plain}
\newtheorem{thm}{Theorem}

\newtheorem{lemma}{Lemma}
\newtheorem{Cor}{Corollary}

\theoremstyle{definition}
\newtheorem{Def}{Definition}

\usepackage{thmtools}

\declaretheoremstyle[
  spaceabove=\topsep, spacebelow=\topsep,
  headfont=\normalfont\bfseries,
  notefont=\mdseries, notebraces={(}{)},
  bodyfont=\normalfont,
  postheadspace=1em,
  qed=\qedsymbol
]{mythmstyle}

\declaretheoremstyle[
  spaceabove=\topsep, spacebelow=\topsep,
  headfont=\normalfont\bfseries,
  notefont=\mdseries, notebraces={(}{)},
  bodyfont=\normalfont,
  postheadspace=1em,
  qed=$\diamond$
]{mythmstyle}
\declaretheorem[style=mythmstyle]{remark}

\usepackage{enumitem}
\setlist[enumerate]{leftmargin = 9pt}
\setlist[itemize]{leftmargin = 7.5pt}

\begin{document}

\allowdisplaybreaks

\sloppy
\title{Improved Approximation of Storage-Rate Tradeoff for Caching with Multiple Demands} 
\author{Avik Sengupta, \IEEEmembership{Student Member,~IEEE}, Ravi Tandon, \IEEEmembership{Member,~IEEE}
\thanks{A. Sengupta is with Wireless@VT and The Hume Center for National Security and Technology, Department of Electrical and Computer Engineering, Virginia Tech, Blacksburg, VA $24060$ USA. Email: aviksg@vt.edu. R. Tandon is with the Department of Electrical and Computer Engineering, University of Arizona, Tucson, AZ $85721$ USA. Email: tandonr@email.arizona.edu.}
\thanks{This work was presented in part at the IEEE International Symposium of Information Theory (ISIT), Hong Kong, June $2015$ and at the Information Theory and Applications Workshop (ITA), UCSD, Feb $2015$.}}
\maketitle 

\begin{abstract}
Caching at the network edge has emerged as a viable solution for alleviating the severe capacity crunch in modern content centric wireless networks by leveraging network load-balancing in the form of localized content storage and delivery. In this work, we consider a cache-aided network where the cache storage phase is assisted by a central server and users can demand multiple files at each transmission interval. To service these demands, we consider two delivery models - $(1)$ centralized content delivery where user demands at each transmission interval are serviced by the central server via multicast transmissions; and $(2)$ device-to-device (D2D) assisted distributed delivery where users multicast to each other in order to service file demands. For such cache-aided networks, we present new results on the fundamental cache storage vs. transmission rate tradeoff. Specifically, we develop a new technique for characterizing information theoretic lower bounds on the storage-rate tradeoff and show that the new lower bounds are strictly tighter than cut-set bounds from literature. Furthermore, using the new lower bounds, we establish the optimal storage-rate tradeoff to within a constant multiplicative gap. We show that, for multiple demands per user, achievable schemes based on repetition of schemes for single demands are order-optimal under both delivery models.

\end{abstract}

\section{Introduction}
The dynamics of traffic over wireless networks has undergone a paradigm shift to become increasingly content centric with high volume multimedia content (e.g., video) distribution holding precedence. Therefore, efficient utilization of network resources is imperative in such networks for improving capacity. With the proliferation of cheap storage at the network edge (e.g. at user devices and small cell base stations), caching has emerged as an important tool for facilitating efficient load balancing and maximal resource utilization for future $5$G wireless networks \cite{bastug}. Parts of popular files are pre-stored at the edge caches such that at times of high network load, the local content can be leveraged to reduce the over-the-air transmission rates. Caching and complimentary file delivery in wireless networks has been the subject of a wealth of recent research as evidenced by the results in \cite{Maddah-Ali, Maddah-Ali-decentralized, Maddah-Ali-nonuniform,Maddah-Ali-Online,aviksg-tifs,aviksg-gws,aviksg-iswcs,Molisch-onecache,ISWCS_Ji,fund_ji,diggavi_mlc,diggavi_hetnet,Motahari_multi-server_arxiv,improve_fund,gunduz_ach,piantanida_ach}. Caching generally works in two phases - $(a)$ the \textit{cache storage phase} where parts of popular content is placed in users' cache memories by a central server e.g., an LTE eNodeB in modern cellular networks and $(b)$ the \textit{file delivery phase}, where requested content is delivered by exploiting local cache storage. Cache placement happens over a much larger time-scale than the file request and delivery phase or a \textit{transmission interval}, and needs to be agnostic to user demands. The fundamental tradeoff in such cache-aided systems is between the cache storage and the delivery rate.

Recently, Maddah-Ali and Niesen \cite{Maddah-Ali, Maddah-Ali-decentralized, Maddah-Ali-nonuniform,Maddah-Ali-Online} showed that by jointly designing the storage and delivery phases, order-wise improvement in the delivery rate can be achieved for any given size of cache storage for the case when users demand only one file at every transmission interval. The proposed  schemes extract a \textit{global caching gain}, in addition to the traditional \textit{local caching gain}, by distributing common content across users' caches and subsequently designing \textit{centralized coded multicast transmissions} which leverage this shared content to reduce delivery rates. The authors used cut-set based arguments to derive an information theoretic lower bound on the optimal storage-rate tradeoff and characterized it to within a constant multiplicative factor of $12$ for worst-case user demands under uniform file popularity. An new lower bound as well as an improved characterization of the optimal storage-rate tradeoff to within a factor of $8$ was presented in our previous work in \cite{aviksg-ISIT}. The case when users demand multiple files at each transmission interval was initially studied in \cite{Ji_Ldem,Ji_Ldem_rand} for the case of worst-case as well as random user demands. The authors in \cite{Ji_Ldem} also proposed a cut-set lower bound for this setting.  

In contrast to the centralized delivery model, a distributed device-to-device (D2D) assisted delivery model was studied in \cite{fund_ji} whereby the delivery phase was relegated to the users instead of a centralized server in order to further reduce backhaul load. The main difference between the centralized content delivery studied in \cite{Maddah-Ali,Maddah-Ali-decentralized} and the D2D-assisted delivery studied in \cite{fund_ji} is the \textit{distributed nature of multicast transmissions}. In the centralized delivery model of Maddah-Ali and Niesen, the multicast can be any arbitrary function of all the files in the library. Instead, for D2D-assisted delivery, the outgoing multicast from each user \textit{can only depend} on the local cache content of that device. In \cite{fund_ji}, Ji et.al. presented new storage/delivery mechanisms for D2D-assisted delivery for the case when each user demands a single file at every transmission interval. The results in \cite{fund_ji} show that even for D2D-assisted delivery, when the devices can use inter-device coded multicast transmissions to satisfy the demands of other users, order-wise improvements in terms of delivery rate can be achieved as compared to uncoded delivery. The authors also presented a cut-set based lower bound on the storage-rate tradeoff. In our prior work in \cite{aviksg-ITA}, we improved on the  cut-set bound and showed that the achievable scheme in \cite{fund_ji} is within a constant multiplicative factor of $8$ from the optimal by leveraging the new bounds. However, the general case when each user can demand \textit{multiple files} at each transmission interval with D2D-assisted delivery has not been considered in literature. 
 
\n \textbf{Main Contributions:} The main contributions of the paper are summarized as follows.
\begin{itemize} 
\item We develop a \textit{new technique} for characterizing information theoretic lower bounds on the storage-rate trade-off for cache-aided systems under centralized and D2D-assisted content delivery for the general case when users can demand multiple files at each transmission interval. 
\item The new lower bounds are shown to be generally tighter than the cut-set bounds in \cite[Theorem 2]{Maddah-Ali} and \cite[Theorem 2]{Ji_Ldem} for all values of problem parameters. The proposed technique also yields the first known converse for the case of D2D-assisted delivery when each user demands multiple files at each transmission interval. 
\item Using the new lower bounds we show that repetitive use of the achievable schemes for single file demands per user, proposed in \cite[Theorem 1]{Maddah-Ali} and \cite[Theorem 1]{fund_ji}, is order-optimal for the case of multiple demands and can characterize the optimal storage-rate tradeoff to within a constant multiplicative factor of $11$ for centralized delivery and $10$ for D2D-assisted delivery, improving upon known results. 
\end{itemize}


\n \textbf{Notation:} For any two integers $a$, $b$ with $a\leq b$, we define $[a:b] \triangleq \{a,a+1,\ldots,b\}$. $b\in [a,c]$ denotes $a \leq b \leq c$ and $b\in (a,c]$ denotes $a<b\leq c$. $Y_{[a:b]}$ denotes the set of random variables $\left\{Y_i : i = [a:b]\right\}$ and $Y_{[a,b]}$ denotes the set $\{Y_i: i = a,b\}$. $\mathbb{N}^+$ denotes the set of positive integers; the function $(x)^+ = \max\{0,x\}$; $\left\lceil x\right\rceil$, $\left\lfloor x\right\rfloor$ are the ceil, floor functions respectively. 

\section{System Model And Preliminary Results} \label{sec:sysmodel}
In this section, we introduce the system model for file storage and delivery in cache-aided systems. We then present achievable schemes for the case of multiple file demands in each transmission interval which are based on repetitions of schemes for single file demands per user.
\vspace{-25pt}\subsection{System Model}
We consider a cache-aided network (see Fig. \ref{fig:sysmod}) with $K$ users and a library of $N$ files, $F_{[1:N]}$, where each file is of size $B$ bits, for $B\in \mathbb{N}^+$. Formally, the files $F_n$ are i.i.d. and distributed as:
\begin{align}
F_n \sim \text{Unif}\{1,2,\ldots, 2^{B}\}, ~~\forall n \in [1:N].
\end{align}
\begin{figure*}[t]
\centering 
\subfigure[]{
\includegraphics[width=2.85in,height=2.8in]{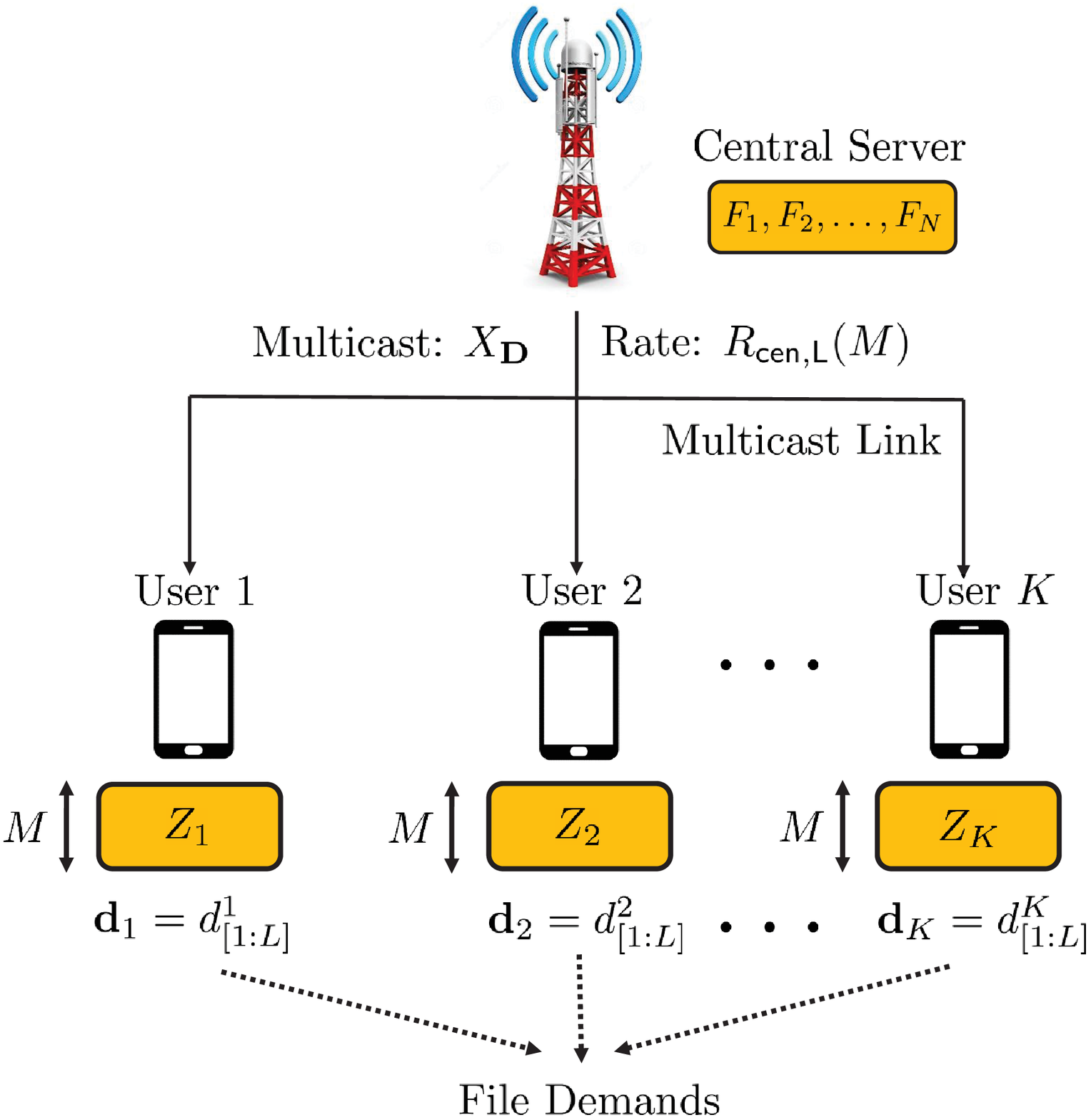}
\label{fig:fsys}
}\hspace{25pt}
\subfigure[]{
\includegraphics[width=3.1in,height=2.8in]{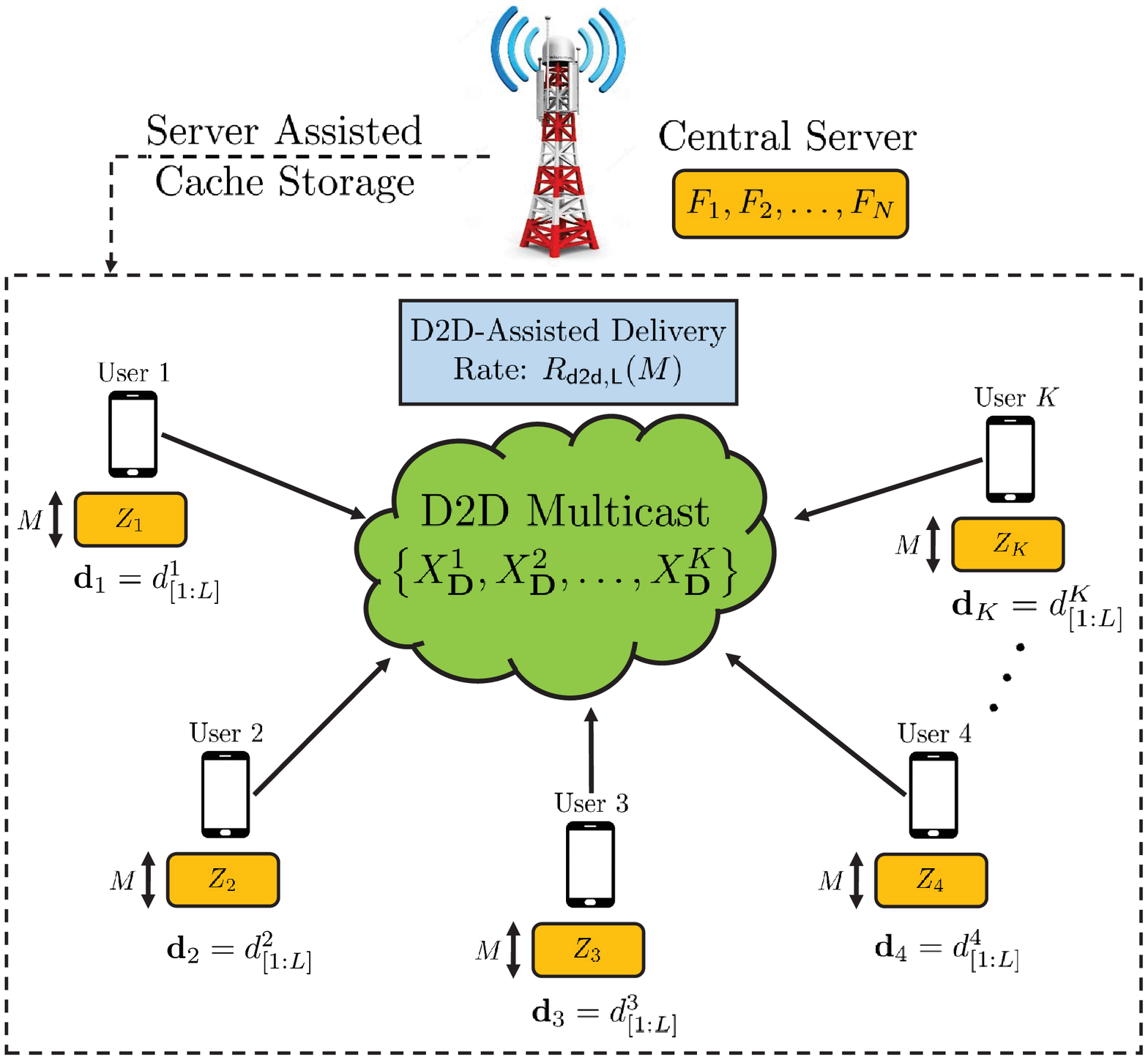}  
\label{fig:dsys}
}\vspace{-1pt}
\caption{ System Model for cache-aided network with $(a)$ centralized content delivery where the requested content is delivered via multicast transmission by the central server; and $(b)$ D2D-assisted content delivery where each device multicasts to all the other devices using the contents placed in the device cache by the central server.} \vspace{-10pt}
\label{fig:sysmod}
\vspace{-5pt}
\end{figure*}
Next, we define the key operational phases and the related performance metric for content storage and delivery in cache-aided systems. 
\begin{Def}[\textit{Cache Storage}]\label{def:store}
The cache storage phase consists of $K$ caching functions which map the files $F_{[1:N]}$ into the cache content
\begin{align}
Z_{k}\triangleq \phi_{k}\Big(F_{[1:N]}\Big),
\end{align}
for each user $k\in [1:K]$. For cache-aided systems with centralized content delivery the cache storage constraint is such that $H\left(Z_k\right)\in [0,MB]$\footnote{Here $H\left(Z_k\right)$ denotes the entropy of the content $Z_k$ stored in the cache of user $k\in[1:K]$ and represents the total size of $Z_k$ in bits i.e., the cache can store at most $M$ files of size $B$ bits each.}. For the case of D2D-assisted delivery, an additional storage constraint is that all caches should be collectively capable of storing the entire library $F_{[1:N]}$ i.e., $KM\geq N$ and $H\left(Z_k\right)\in [NB/K,MB]$\footnote{The lower bound follows from the fact that each cache needs to store at least $N/K$ files.}. The cache placement phase generally occurs over a larger time-scale encompassing multiple user demand phases or \textit{transmission intervals}. As a result, the caching functions are agnostic to user demands.
\end{Def}

\begin{Def}[\textit{File Delivery}]\label{def:del}
The file delivery phase occurs in each transmission interval in response to user demands with each user requesting $L\in[1:N]$ files. The user demands are denoted by $\mathbf{D}= \mathbf{d}_{[1:K]}$, where each users' demand vector consists of $L$ distinct files $\mathbf{d}_k = d^k_{[1:L]} \in [1:N]$ for $k\in[1:K]$. For the case of centralized delivery, the central server uses $N^{KL}$ encoding functions to map the library of files $F_{[1:N]}$ to the multicast transmission 
\begin{align}\label{eq:multicast}
X_{\mathbf{D}}\triangleq \psi_{\mathbf{D}}\big(F_{1}, \ldots, F_{N}\big), 
\end{align}
over the shared link with a rate not exceeding $RB$ bits i.e., $H\left(X_{\mathbf{D}}\right)\leq RB$. For D2D-assisted delivery, the encoding function $\psi_{\mathbf{D}}$ is composed of $K$ functions, $\psi^{k}_{\mathbf{D}}$, one for each user. The $K$ users encode the contents of their respective caches into a composite D2D multicast transmission
\begin{align}
X_{\mathbf{D}} = \Big\{\left( X^{1}_{\mathbf{D}},X^{2}_{\mathbf{D}},\ldots,X^{K}_{\mathbf{D}}\right): X^{k}_{\mathbf{D}} = \psi^{k}_{\mathbf{D}}\left(Z_k\right), ~\forall k\in [1:K]\Big\}.
\end{align}
Each multicast transmission has a rate $\sum_{k=1}^K H\left(X^k_{\mathbf{D}}\right)\leq RB$.
\end{Def}

\begin{Def}[\textit{File Decoding}]\label{def:dec}
Once the multicast transmission is received, $KN^{KL}$ decoding functions map the received signal $X_{\mathbf{D}}$ and the local cache content $Z_{k}$ to the estimates
\begin{align}
\widehat{F}_{\mathbf{d}_{k}}\triangleq \mu_{\mathbf{D}, k} \Big(X_{\mathbf{D}}, Z_{k}\Big),
\end{align}
of the $L$ requested files $F_{\mathbf{d}_{k}}$ for user $k\in [1:K]$. The probability of error is defined as 
\begin{align}
P_{e}~\triangleq ~\max_{\mathbf{D},~ k\in [1:K], ~d\in\mathbf{d}_k} \mathbb{P}\left(\hat{F}_{d}\neq F_{d}\right),
\end{align}
i.e., the worst-case probability of error evaluated over all possible demand vectors and across all users for any number of per user demands $L$. 
\end{Def}

\begin{Def}[\textit{Storage-Rate Tradeoff}]\label{def:tradeoff}
The storage-rate pair $(M,R_{\mathsf{cen,L}})$ for centralized delivery or $(M,R_{\mathsf{d2d,L}})$ for D2D-assisted delivery is \textit{achievable} if, for any $\epsilon > 0$, there exists a caching and delivery scheme, for which $P_e \leq \epsilon$, where $\epsilon$ is an arbitrarily small constant. The optimal storage-rate tradeoffs are defined as
\begin{align}
&R^*_{\mathsf{cen,L}}(M) \triangleq \inf\left\{R_{\mathsf{cen,L}} : (M,R_{\mathsf{cen,L}})~ \text{is achievable} \right\};\\
&R^*_{\mathsf{d2d,L}}(M) \triangleq \inf\left\{R_{\mathsf{d2d,L}} : (M,R_{\mathsf{d2d,L}})~ \text{is achievable} \right\}.
\end{align}
\end{Def}

\subsection{Preliminary Results}
In this section, we present existing achievability results which yield upper bounds on the optimal storage-rate tradeoff for cache-aided systems under centralized as well as D2D-assisted delivery for the case of $L(\geq 1)$ demands per user. 
\subsubsection{Centralized Delivery with Multiple Demands}
An achievable scheme for caching with centralized delivery was first proposed in \cite{Maddah-Ali} for the case of single $(L=1)$ user requests. An extension to the case when each user can make multiple ($L>1$) demands at any given transmission interval is given by the following lemma. 

\begin{lemma}\label{lem:ldem_ach}
For any $N$ files and $K$ users, with each user having cache storage of $M\in\frac{Nt}{K}$ files for any $t\in [0:K]$, an achievable content delivery rate which upper bounds the optimal rate is given by:
\begin{align}\label{eq:ldem_ach}
&R^*_{\mathsf{cen,L}}(M) \leq R_{\mathsf{cen,L}}(M) = KL\left(1 - \frac{M}{N}\right) \min \left(\frac{1}{1 + {KM}/{N}}, \frac{N}{KL} \right),
\end{align}
for the case when each user requests any $L \in [1:N]$ files at every transmission interval.
\end{lemma}
\begin{proof}
 The delivery rate in \eqref{eq:ldem_ach} can be achieved by a strategy which repeats $L$ times, the coded multicast delivery scheme proposed in \cite[Theorem $1$]{Maddah-Ali}. The second term inside the $\min(\cdot)$ function is derived from the unicasting of $\min\{N,KL\}$ files.
\end{proof}
  
\subsection{D2D-assisted Delivery with Multiple Demands per Device}
For the case of D2D-assisted delivery, Ji et. al. proposed an order-optimal caching and delivery scheme in \cite{fund_ji} for case of single $(L=1)$ user demands. An extension to the case of multiple $(L>1)$ demands per user, is given by the following lemma.
\begin{lemma}\label{lem:d2d_ldem}
For any $N$ files and $K$ users, each having storage size $M\in\frac{Nt}{K}$ files for any $t\in [0:K]$ with $KM\geq N$, an achievable rate for D2D-assisted content delivery is given by
\begin{align}\label{eq:d2d_ldem_ub}
R_{\mathsf{d2d,L}}(M) \leq \min\left\{\frac{LN}{M}\left(1 - \frac{M}{N}\right),N \right\},
\end{align}
for the case when each user requests any $L \in [1:N]$ files at every transmission interval.
\end{lemma} 
\begin{proof}
The delivery rate in \eqref{eq:d2d_ldem_ub} can be achieved by a strategy which repeats $L$ times, the coded multicast delivery scheme proposed in \cite[Theorem $1$]{fund_ji}. The second term inside the $\min(\cdot)$ function is derived from the multicasting of all $N$ files, which is possible since the storage constraint for D2D-assisted delivery ensures that $KM\geq N$.
\end{proof}

In \cite{Ji_Ldem}, Ji et. al presented a graph-coloring based index coded delivery scheme which showed that coding across files as well as demands can improve the centralized delivery rate compared to the approach in Lemma \ref{lem:ldem_ach}, while D2D-assisted delivery schemes specifically for multiple $(L>1)$ demands has not been studied in literature. In this work, we address the following question - \textit{are the repetition based schemes order-optimal, thereby foregoing the need for more complex approaches?} An answer in the affirmative is provided in Section \ref{sec:convfund}, where we leverage the proposed lower bounds to prove the order-optimality of the schemes presented above.

\section{Main Results and Discussion}\label{sec:convfund}
In this section, we present new converse bounds for centralized and D2D-assisted content delivery in cache-aided networks with multiple ($L\geq 1$) demands per user.

\subsection{Centralized Content Delivery}
We next present our first main result which gives a new lower bound on the optimal storage-rate tradeoff for cache-aided systems with centralized content delivery.

\begin{thm}\label{th:ldem}
For any  $N$ files and $K$ users, each having a cache size of $M\in[0,N]$, the optimal centralized content delivery rate $R_{\mathsf{cen,L}}^*(M)$ is lower bounded as
\begin{align}\label{eq:ldem}
&R_{\mathsf{cen,L}}^*(M) \geq \max_{\substack{s \in [1:\min\{\lceil N/L \rceil, K\}] ,\\ \ell \in \left[ 1 : \left\lceil {N/(Ls)} \right\rceil \right]}}~ \frac{1}{\ell}\left\{ N - sM - \frac{\mu (N - L\ell s )^{+}}{s+\mu} - (N - KL\ell )^{+}\right\},
\end{align}
for the case when each user demands $L$ $\in$ $[1:N]$ files at every transmission interval. The parameter $\mu =\big(\min\big(\left\lceil N/(L\ell) \right\rceil, K\big) -s\big), ~ \forall s,\ell~$. 
\end{thm}
The proof of Theorem \ref{th:ldem} is given in Appendix \ref{ap:ldem}. The expression in Theorem \ref{th:ldem} has two parameters, namely $(i)$ the parameter $s$, which is related to the number of user caches; and $(ii)$ the parameter $\ell$, which is related to multicast transmissions. Compared to the cut-set bounds presented in \cite[Theorem 2]{Ji_Ldem}, the additional parameter $\ell$ adds further flexibility to the lower bound expression and accounts for file decoding through the interaction of caches and transmissions, yielding a generally tighter lower bound for the case of centralized content delivery with multiple demands per user. The main difference between the cut-set bound and the proposed lower bound is based on the fact that the new bounds better utilize the possible correlation between caches by carefully bounding the joint and conditional entropy of subsets of cache storages by utilizing Han's inequality on subsets (see Section \ref{sec:exm_cen} for more details). The cut-set based lower bound of \cite[Theorem 2]{Ji_Ldem} is tight only for very large values of cache size $M$. As shown in the sequel, for such values of $M$, the proposed bound yields the cut-set bounds for specific choices of $s$ and $\ell$ and is generally tighter for all other values. This is illustrated in Fig. \ref{fig:ldemands} where we show that the proposed bound in strictly tighter than the cut-set bound.
\begin{figure*}[!t]
\centering 
\subfigure[]{
\includegraphics[width=3.5 in,height=2.75in]{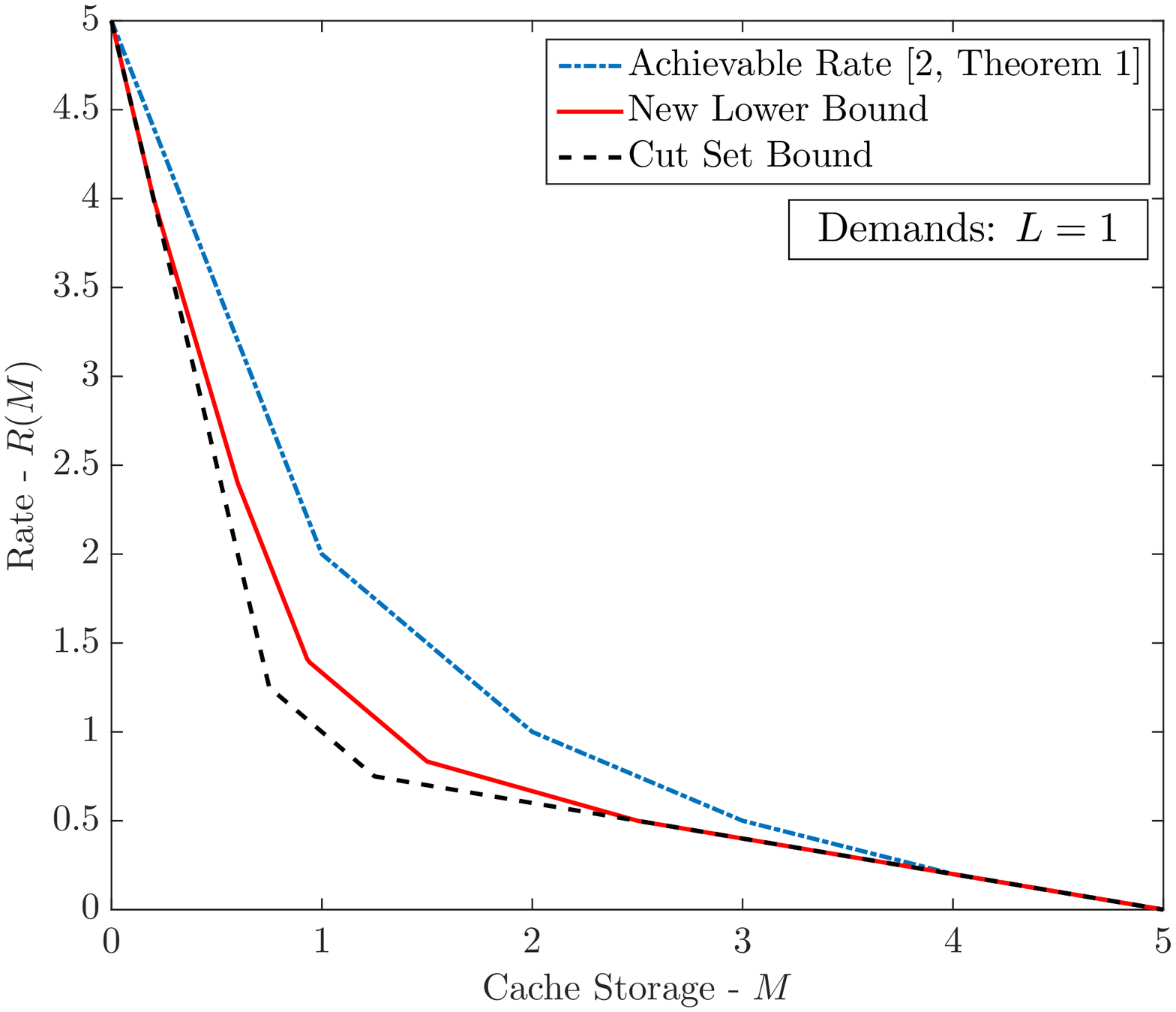}
\label{fig:nk55}
}\hspace{-25pt}
\subfigure[]{
\includegraphics[width=3.5 in,height=2.75in]{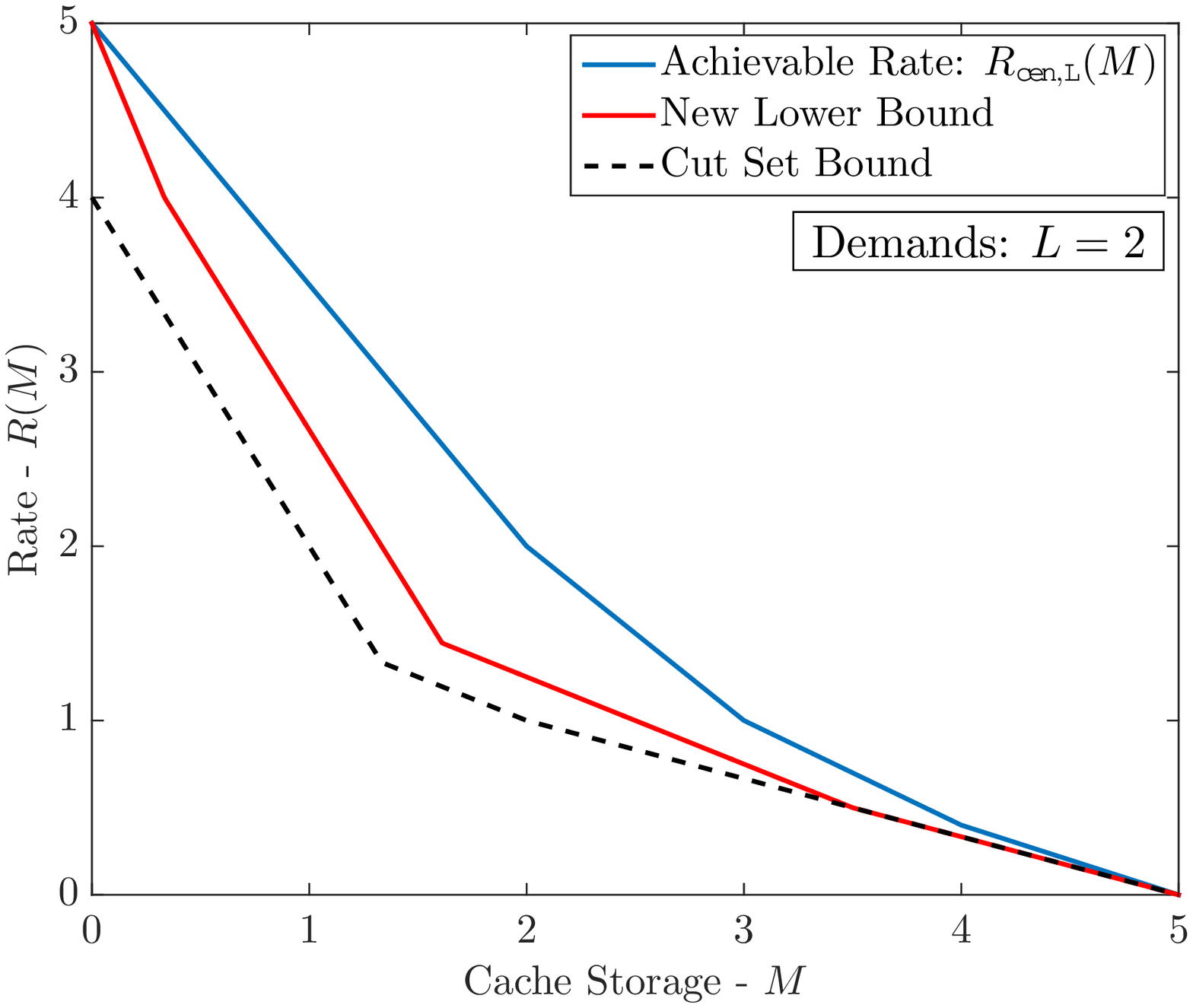}
\label{fig:ldemands}
}
\vspace{-10pt}
\caption{ Storage-rate trade-off for centralized content delivery with $N=K=5$ and ${(a)}$ $L=1$ demand per user; and ${(b)}$ $L=2$ demands per user.} \vspace{-10pt}
\end{figure*}

 We next present our second main result which shows that an improved approximation of the optimal storage-rate tradeoff can be obtained by use of the proposed lower bound. 
\begin{thm}\label{th:ldem_gap}
For any $N$ files and $K$ users, each with a cache size of $M \in [0,N]$, and each user requesting $L(\leq N)$ files at each transmission interval, we have:
\begin{align}\label{eq:ldem_gap}
&\frac{R_{\mathsf{cen,L}}(M)}{R^*_{\mathsf{cen,L}}(M)} \leq 11.
\end{align}
\end{thm}
 The proof of Theorem \ref{th:ldem_gap} is provided in Appendix \ref{ap:ldem_gap}. This result improves on the gap of $18$ between the achievable scheme and the cut-set bound in \cite[Theorem 2]{Ji_Ldem}. Furthermore, the result shows that performing $L$ repetitions of the scheme for a single demand (as in Lemma \ref{lem:ldem_ach}) is in fact order-optimal, thereby precluding the need for more complex schemes as in \cite{Ji_Ldem}.
 
\begin{Cor}\label{cor:convfund}
For any $N$ files and $K$ users, each having a cache size of $M \in [0, N]$, the optimal centralized content delivery rate $R_{\mathsf{cen}}^*(M)$ for the case when each user requests $L=1$ file at every transmission interval, is lower bounded by:
\begin{align}\label{eq:convfund}
&R^*_{\mathsf{cen}}(M) \geq \max_{\substack{s \in [1:K] , ~\ell \in \left[ 1 : \left\lceil {{N}/{s}} \right\rceil \right]}} \frac{1}{\ell}\left\{ N - sM - \frac{\mu (N - \ell s )^{+}}{s+\mu} - (N - K\ell )^{+}\right\},
\end{align}
where $\mu =\big(\min\big(\left\lceil N/\ell\right\rceil, K\big) -s\big), ~~ \forall s,\ell~$. 
\end{Cor}
Corollary \ref{cor:convfund} follows by setting $L=1$ in Theorem \ref{th:ldem} and was originally presented in \cite{aviksg-isit15}. The new bounds strictly improve on the cut-set lower bounds presented in \cite[Theorem 2]{Maddah-Ali} as shown in Fig. \ref{fig:nk55}. Using \eqref{eq:convfund}, the approximation of the optimal storage-rate tradeoff can be improved as follows.

\begin{thm} \label{th:gapfund}
Let $R_{\mathsf{cen}}(M)$ be the achievable rate of the centralized caching scheme given in \cite[Theorem 1]{Maddah-Ali}. Then, for any $K$ users, $N$ files, and user cache storage in the range $M \in [0,N]$, we have:
\begin{align}
 \frac{R_{\mathsf{cen}}(M)}{R^*_{\mathsf{cen}}(M)} \leq 8.
\end{align}
\end{thm}
The proof of Theorem \ref{th:gapfund} is provided in Appendix \ref{ap:gapfund}. The result improves on the gap of $12$ yielded by the cut-set bound in \cite[Theorem 3]{Maddah-Ali}\footnote{The results presented in this paper also hold for the case of \textit{decentralized cache placement} as in \cite{Maddah-Ali-decentralized} since the converse makes no assumption on the nature of content placement.}.


\subsection{D2D-Assisted Content Delivery}
In this section, we consider the case of D2D-assisted content delivery with each user demanding multiple files in each transmission interval. The next theorem presents our main result which gives the first-known lower bound on the optimal storage-rate tradeoff.

\begin{thm}\label{th:d2d_ldem}
For any $N$ files and $K$ users, each having a cache size of $M\in[N/K,N]$, the optimal D2D-assisted content delivery rate $R_{\mathsf{d2d,L}}^*(M)$ is lower bounded as
\begin{align}\label{eq:d2d_ldem}
R^*_{\mathsf{d2d,L}}(M) \geq \max_{\substack{s \in [1:\min\{\lceil N/L \rceil, K\}] ,~\ell \in \left[ 1 : \left\lceil {\frac{N}{Ls}} \right\rceil \right]}} \left\{  \frac{N - sM - \frac{\mu}{s+\mu} (N - L\ell s )^{+}}{\ell\left(\frac{K-s}{K}\right)}\right\},
\end{align}
for the case when each user demands $L\in[1:N]$ files at each transmission interval. The parameter $\mu =\left(\min\left(\left\lceil N/(L\ell)\right\rceil, K \right) -s\right), ~ \forall s,\ell~$. 
\end{thm}

The proof of Theorem \ref{th:d2d_ldem} is presented in Appendix \ref{ap:d2d_ldem}. Similar to Theorem \ref{th:ldem}, the parameters $s$ and $\ell$ yield a family of lower bounds by exploiting the correlation between the caches and transmissions by use of Han's Inequality. Fig. \ref{fig:d2d_ldemands} shows the lower bound in \eqref{eq:d2d_ldem} and the upper bound $R_{\mathsf{d2d,L}}(M)$ given in \eqref{eq:d2d_ldem_ub}. Leveraging the proposed lower bound, we present our second main result in the following theorem.

\begin{thm}\label{th:d2d_ldem_gap}
For any $N$ files and $K$ users, each having a cache size of $M\in[N/K,N]$, and with each user requesting $L(\leq N)$ files at each transmission interval, we have
\begin{align}\label{eq:d2d_ldem_gap}
\frac{R_{\mathsf{d2d,L}}(M)}{R^*_{\mathsf{d2d,L}}(M)} \leq 10.
\end{align}
\end{thm}
The proof of Theorem \ref{th:d2d_ldem_gap} is presented in Appendix \ref{ap:d2d_ldem_gap}. The result shows that a repetition based scheme as outlined in Lemma \ref{lem:d2d_ldem} is in fact order-optimal and yields a constant factor approximation of the storage-rate trade-off for D2D-assisted content delivery with multiple demands per user.

\begin{figure*}[t]
\centering
\subfigure[]{
\includegraphics[width=3.5 in,height=2.75in]{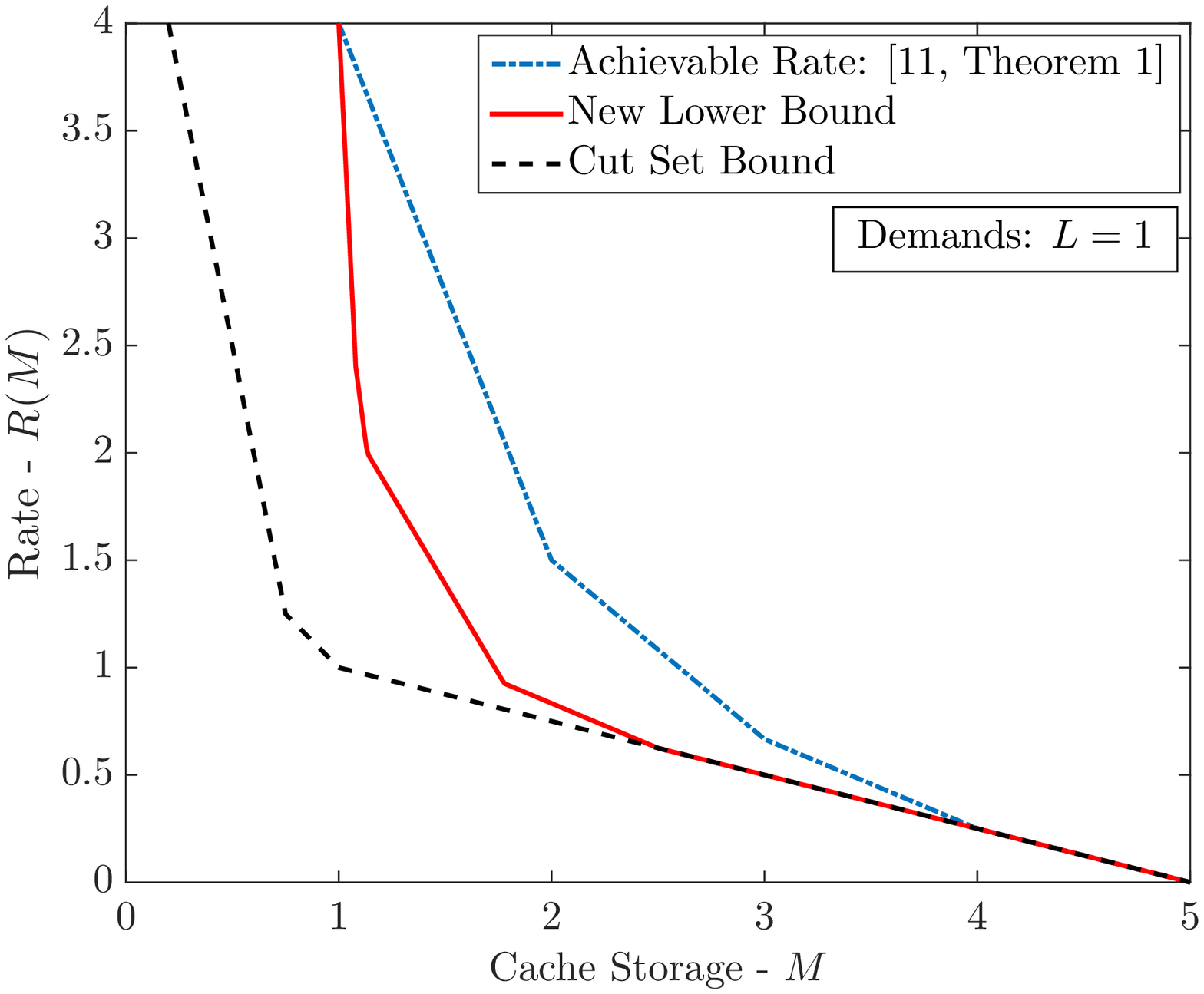}
\label{fig:nkd55}
}\hspace{-25pt}
\subfigure[]{
\includegraphics[width=3.5 in,height=2.75in]{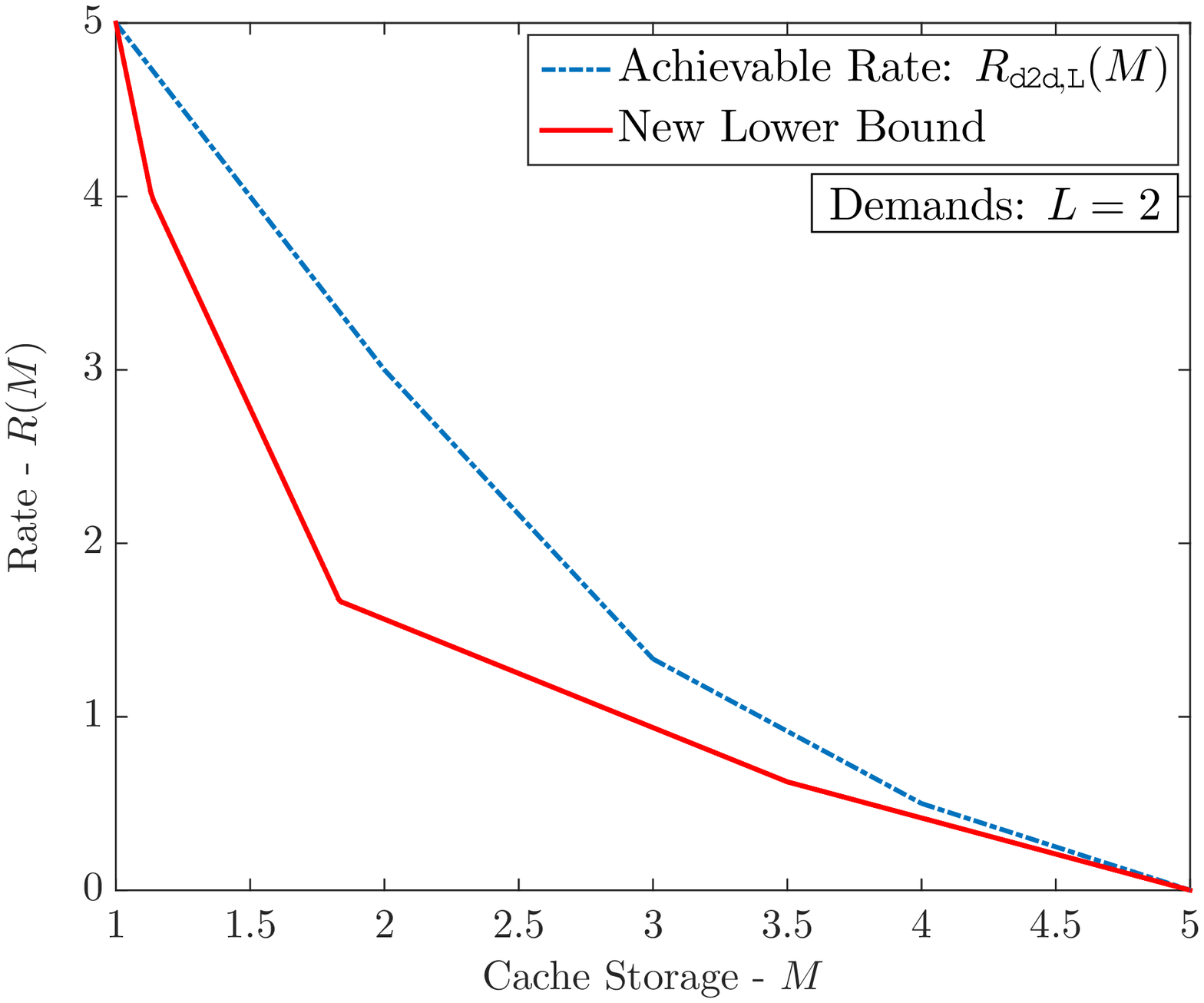}
\label{fig:d2d_ldemands}
}
\vspace{-10pt}
\caption{ Storage-rate tradeoff for D2D-assisted content delivery with $N=K=5$ and $(a)$ $L=1$ demand per user; and $(b)$ $L=2$ demands per user.} \vspace{-10pt}
\end{figure*}

\begin{Cor}\label{cor:convd2d}
For any $N$ files and $K$ users, each having a cache size of $M\in[N/K,N]$, the optimal D2D-assisted content delivery rate $R_{\mathsf{d2d}}^*(M)$, for the case when each user requests $L=1$ file at every transmission interval, is lower bounded by:
\begin{align}\label{eq:convd2d}
&R^*_{\mathsf{d2d}}(M) \geq \max_{\substack{s \in [1:K] , ~\ell \in \left[ 1:\left\lceil {{N}/{s}} \right\rceil \right]}}  \left\{\frac{N - sM - \left(\frac{\mu}{s+\mu}\right)(N - \ell s )^{+}}{\ell\left(\frac{K - s}{K}\right)}\right\}, 
\end{align}
where $\mu =\left(\min\left(\left\lceil N/\ell\right\rceil, K\right)-s\right) ~ \forall s,\ell~$. 
\end{Cor}
 Corollary \ref{cor:convd2d} follows by setting $L=1$ in Theorem \ref{th:d2d_ldem} and was originally presented in \cite{aviksg-ITA}. Compared to the cut-set bound in \cite[Theorem 2]{fund_ji}, we note that the proposed bound in Corollary \ref{cor:convd2d} is always tighter owing to the additional parameter $\ell$ and the factor $(K-s)/K \leq 1$ in the denominator of \eqref{eq:convd2d}. Furthermore, the bound in \cite{fund_ji} is tight only for large  values of device storage size $M$. The new bound is tighter for smaller values of $M$ and yields the existing bound as a special case for large values of $M$. 

Using the lower bound in Corollary \ref{cor:convd2d}, the optimal storage-rate tradeoff for the case of single demands per user can be approximated as follows.

\begin{thm} \label{th:gapd2d}
For any $K \in \mathbb{N}^+$ user devices, $N\in \mathbb{N}^+$ files, and device storage in the range $M \in \left[\frac{N}{K}, N\right]$, we have:
\begin{align}
\frac{R_{\mathsf{d2d}}(M)}{R^*_{\mathsf{d2d}}(M)}~ \begin{cases}~
																	=~ 1 ~~~ &M = {N}/{K}\\
																	~\leq~ 3 ~~~ & M \in \left({N}/{K}, {2}/{3}\right]\\
																	~\leq~ 6 ~~~ &M \in \left({2}/{3},1\right]\\
																	~\leq~ 8 ~~~ &1 \leq M \leq N
																 \end{cases}.
\end{align}
\end{thm}
The proof of Theorem \ref{th:gapd2d} is presented in Appendix \ref{ap:gapd2d}. The result highlights the fact that for the smallest allowable cache size of $M=N/K$, the lower bound in \eqref{eq:convd2d} is tight and yields the achievable rate in \cite[Theorem 1]{fund_ji}. This is also shown in Fig. \ref{fig:nkd55} for the case of $N=K=5$ and $L=1$. 

\begin{figure*}[!t]
\centering \vspace{-10pt}
\subfigure[]{
\includegraphics[width=4.25in,height=2.45in]{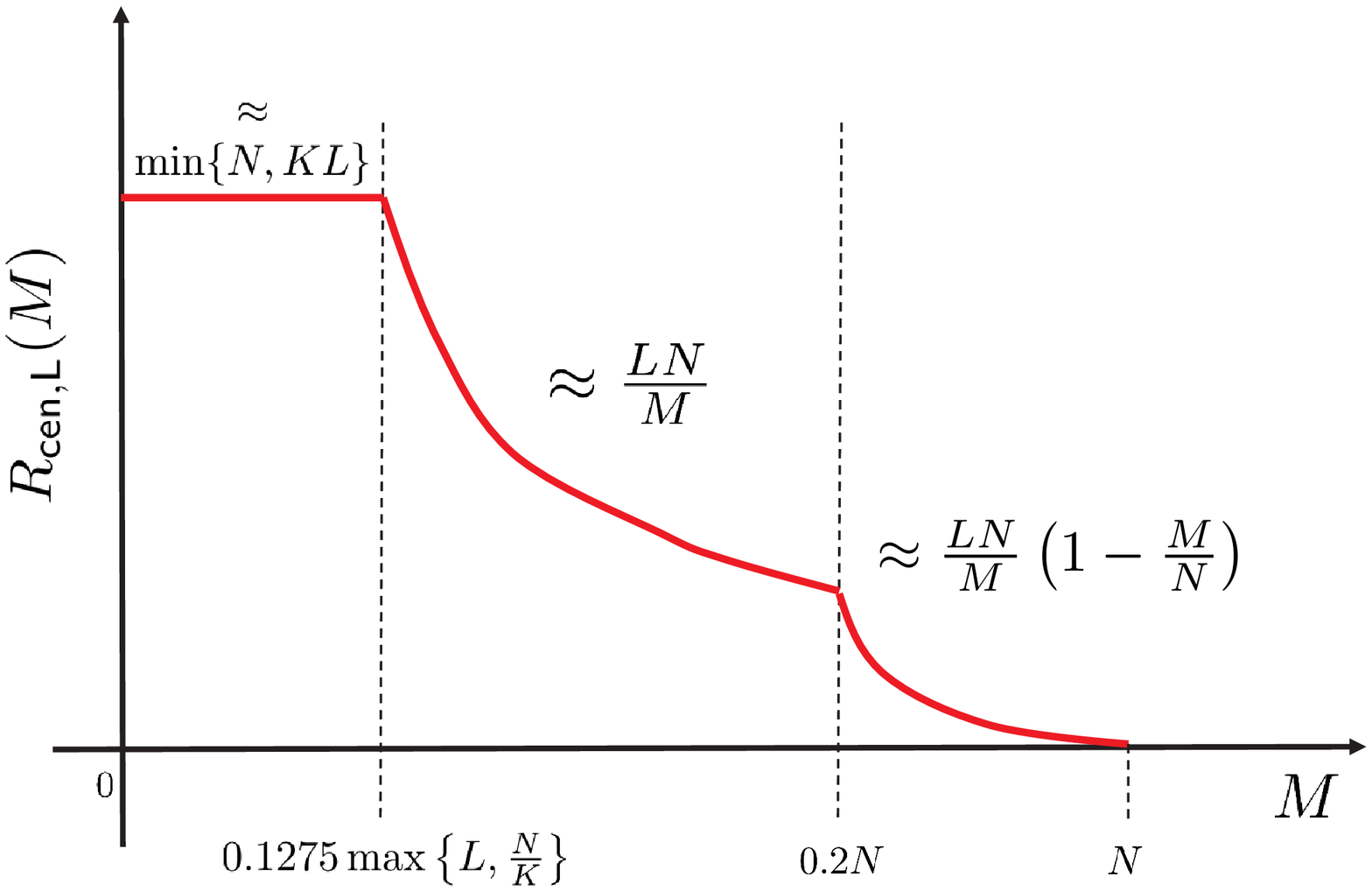}
\label{fig:cenapprox}
}\hspace{-25pt}\\
\subfigure[]{
\includegraphics[width=3.00in,height=2.025in]{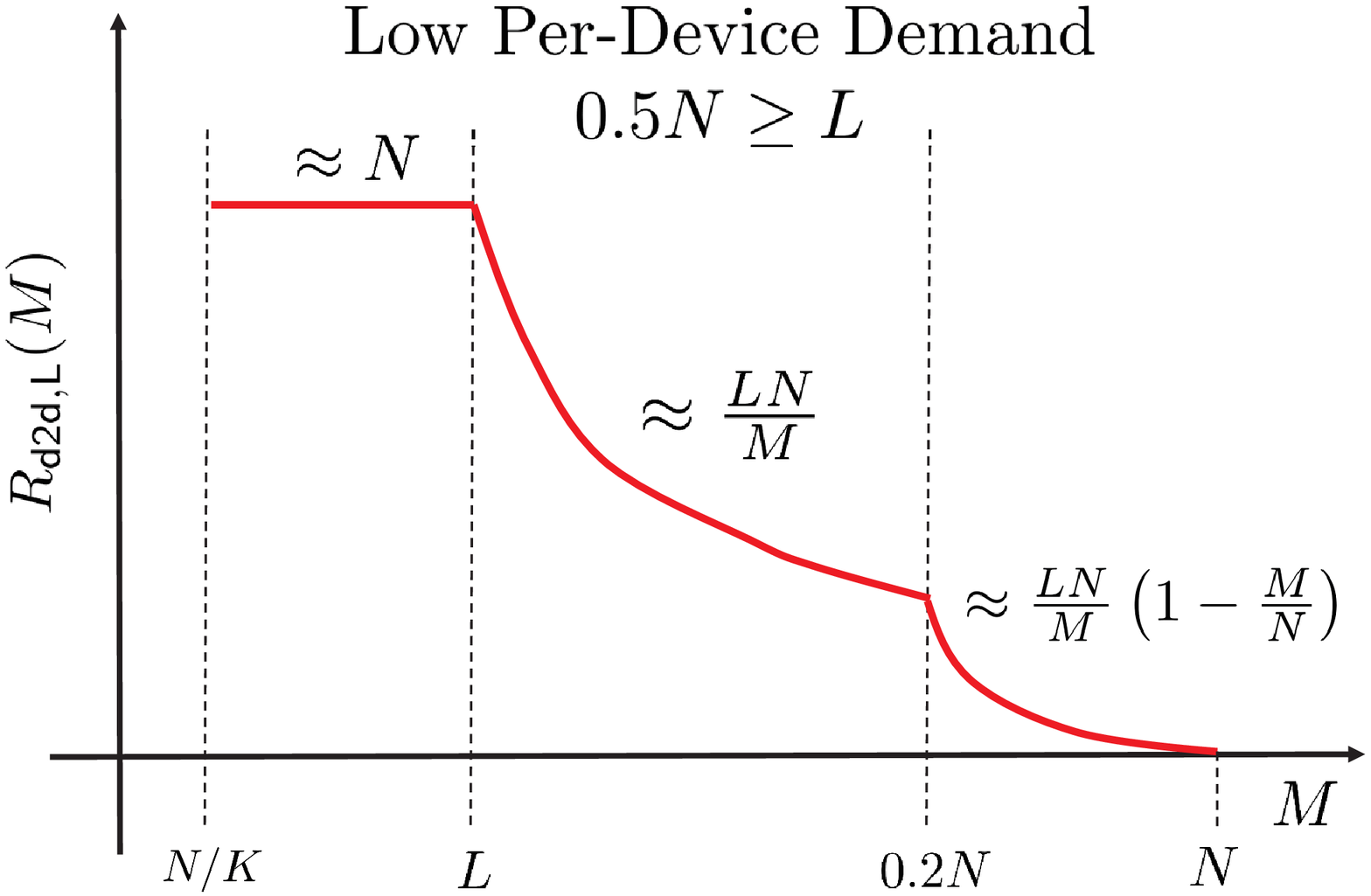}
\label{fig:d2dapprox_ld}
}\hspace{25pt}
\subfigure[]{
\includegraphics[width=3.00in,height=2.025in]{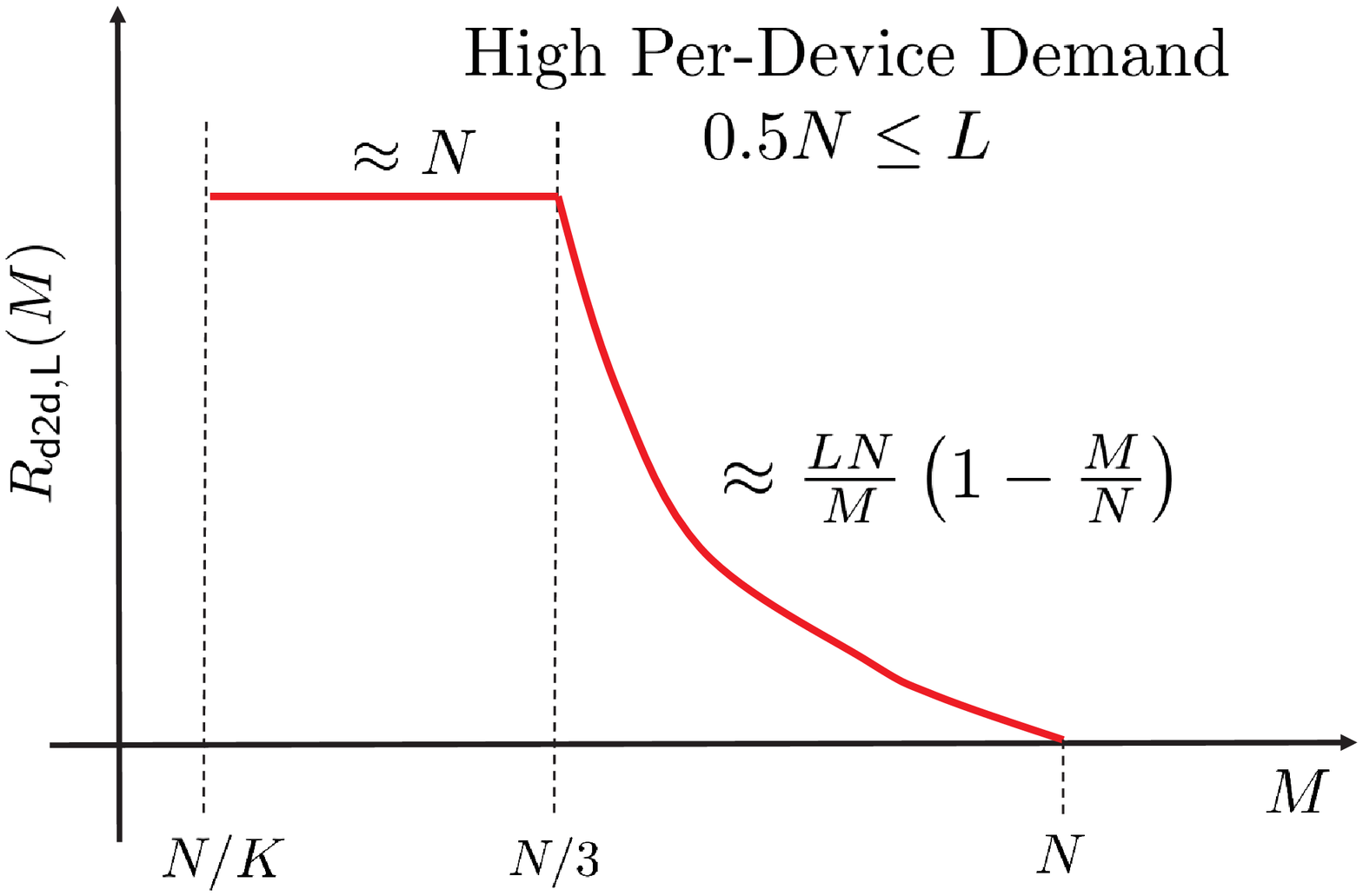}
\label{fig:d2dapprox_hdem}
}
\vspace{-10pt}
\caption{A representation of the order-optimal approximations to the delivery rate for the repetition based schemes for $(a)$ centralized content delivery, which is used in the proof of Theorem \ref{th:ldem_gap}; and $(b)-(c)$ for D2D-assisted content delivery with low and high per-device demands, which are used in the proof of Theorem \ref{th:d2d_ldem_gap}.} \vspace{-10pt}
\label{fig:approx}
\end{figure*}

\begin{remark}
To prove the order-optimality of the repetition based schemes as shown in Theorems \ref{th:ldem_gap} and \ref{th:d2d_ldem_gap}, we use approximations to the achievable rates presented in Lemmas \ref{lem:ldem_ach} and \ref{lem:d2d_ldem}. These approximations are highlighted in Fig. \ref{fig:approx}. For the case of centralized content delivery, three regimes of cache storage are considered and for very low cache storage, it is approximately optimal to unicast all requested files as seen in Fig. \ref{fig:cenapprox}. For higher cache storage, a linear dependance of the rate on $L/M$ is established. For the case of D2D-assisted delivery, we see that when users demand less than half the library, three regimes of cache storage need to be considered, while for the case of high per-device demands, only 2 regimes suffice and for storage as high as a third of the library, it is approximately optimal for all users to broadcast all $N$ files from their local caches. Further details are provided in Appendix \ref{ap:ldem_gap} and \ref{ap:d2d_ldem_gap}.
\end{remark}

\section{Case Studies}
In this section, we present two case studies to illustrate the new techniques used to obtain the lower bounds in Theorems \ref{th:ldem} and \ref{th:d2d_ldem}. For ease of exposition, we consider the special case of $L=1$ since the results easily extend to any $L>1$. We show that our technique yields additional bounds as compared to the cut-set techniques in literature and present discussions behind the principal intuitions in applying our method. To this end we first consider the case of centralized content delivery.

\subsection{Centralized Content Delivery: Intuition Behind Proof of Theorem \ref{th:ldem}}\label{sec:exm_cen}
We consider $N=3$ files, denoted by $A,B,C$ and $K=3$ users, each with a cache storage $M$ files. For the case of $L=1$, Corollary \ref{cor:convfund} yields the following lower bounds for different $s,\ell$.
\begin{align}
&\text{New Lower Bounds:}&3R_{\mathsf{cen}}^* + 6M \geq 8, ~ s = 2, ~\ell = 1;~~4R_{\mathsf{cen}}^* + 2M \geq 5,~ s = 1, ~\ell = 2\label{nlb2}\\
&\text{Cut-Set Bounds:}\hspace{0pt} &R_{\mathsf{cen}}^* + 3M \geq 3, ~ s = 3, ~\ell = 1;~~~3R_{\mathsf{cen}}^* + ~M \geq 3, ~ s = 1, ~\ell = 3\label{cs2}.
\end{align}
The existing lower bounds from \cite[Theorem 2]{Maddah-Ali} are given by (\ref{cs2}). The proposed approach provides the additional bounds in (\ref{nlb2}), thereby yielding tighter lower bounds than \cite[Theorem 2]{Maddah-Ali} as shown in  Fig. \ref{fig:nk33}. Next, we detail the derivation of the first bound in (\ref{nlb2}) highlighting the new aspects and techniques. 

To this end, we consider two consecutive requests $(d_1,d_2,d_3) = (A,B,C)$ and $(d_1,d_2,d_3) = (B,C,A)$.  It is clear that the first  $s=2$ caches $Z_{[1,2]}$ along with two corresponding transmissions $X_{ABC},X_{BCA}$ from the central server  suffice to decode all the $3$ files. We upper bound the entropy of $\ell = 1$ multicast transmission by the optimal rate $R_{\mathsf{cen}}^*$ and use the other transmission's decoding capability with the caches to derive the following bound
\begin{align}
3B & \leq H(Z_{[1,2]}, X_{ABC},X_{BCA}) \leq H(Z_{[1,2]}) + H(X_{ABC},X_{BCA}|Z_{[1,2]})\nonumber\\
	 & \leq 2MB + H(X_{ABC}) + H(X_{BCA}|Z_{[1,2]},X_{ABC})\nonumber\\
	 & \myleq{(a)} 2MB + R_{\mathsf{cen}}^*B + H(X_{BCA}|Z_{[1,2]},X_{ABC},A,B)\nonumber\\
	 & \leq 2MB + R_{\mathsf{cen}}^*B + H(X_{BCA},Z_3|Z_{[1,2]},X_{ABC},A,B)\nonumber\\
	 & \leq 2MB + R_{\mathsf{cen}}^*B + H(Z_3|Z_{[1,2]},X_{ABC},A,B) +  H(X_{BCA}|Z_{[1:3]},X_{ABC},A,B)\nonumber\\
	 & \myleq{(b)} 2MB + R_{\mathsf{cen}}^*B + H(Z_3|Z_{[1,2]},A,B) +  H(X_{BCA}|Z_{[1:3]},X_{ABC},A,B,C)\nonumber\\
	 & \leq 2MB + R_{\mathsf{cen}}^*B + H(Z_3|Z_{[1,2]},A,B) \label{eq:ex2:1},
\end{align}
where step \textsf{(a)} follows from the fact that $Z_{[1,2]}$ along with $X_{ABC}$ can decode files $A,B$ and step \textsf{(b)} follows from the fact that $H(X_{BCA}|Z_{[1:3]},X_{ABC},A,B,C)=0$ since each transmission is a deterministic function of the files. Considering the term $H(Z_3|Z_{[1,2]},A,B)$ in (\ref{eq:ex2:1}), we have:
\begin{align}\label{eq:sym1}
H(Z_3|Z_{[1,2]},A,B) = H(Z_{[1:3]}|A,B) - H(Z_{[1,2]}|A,B).
\end{align}
Using (\ref{eq:sym1}) in  (\ref{eq:ex2:1}), we have:
\begin{align}
3B & \leq 2MB + R_{\mathsf{cen}}^*B +  H(Z_{[1:3]}|A,B) - H(Z_{[1,2]}|A,B).\label{eq:ex2:2}
\end{align}
\begin{figure*}[!t]
\centering 
\subfigure[]{
\includegraphics[width=3.5in,height=2.75in]{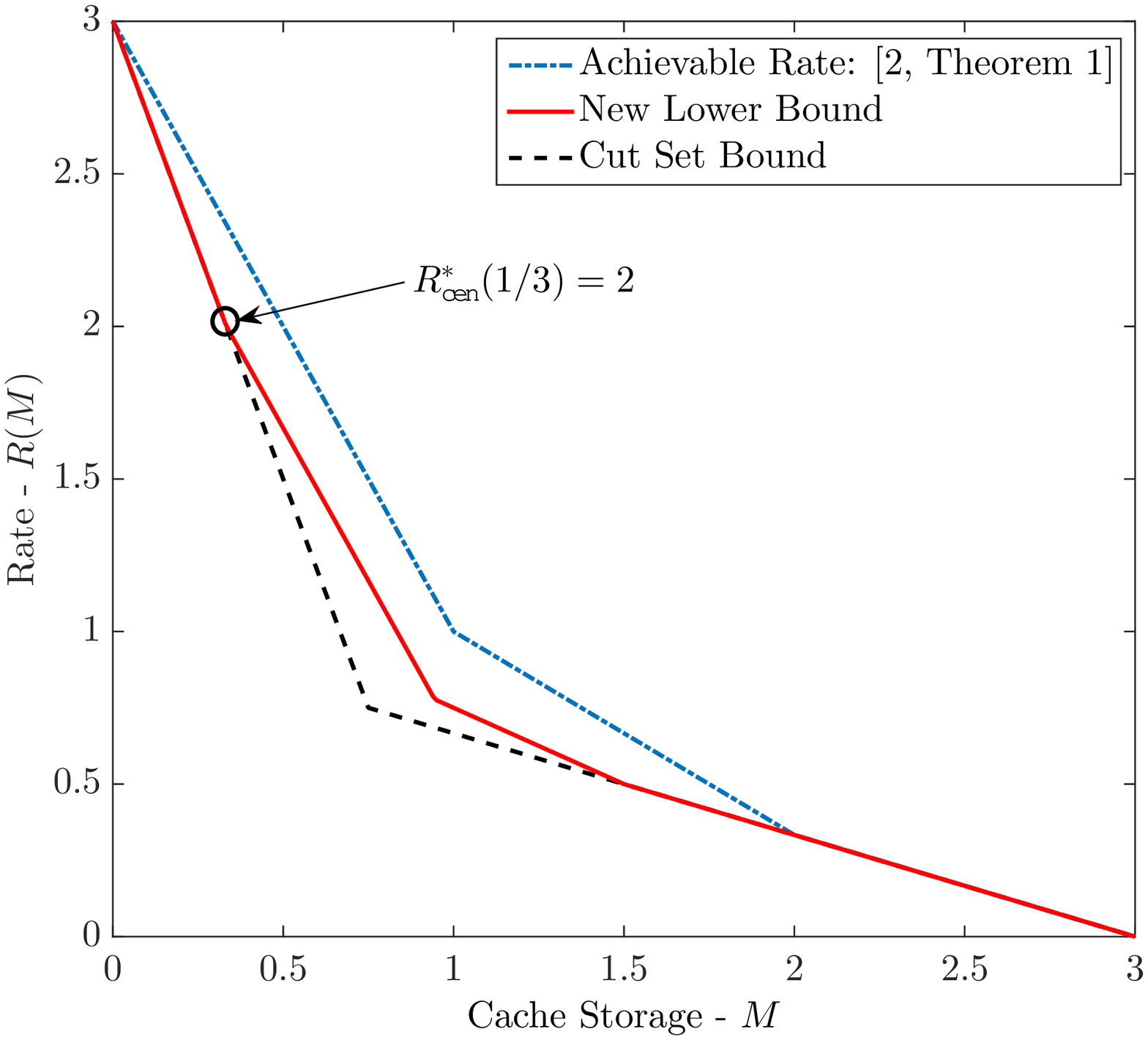}
\label{fig:nk33}
}\hspace{-25pt}
\subfigure[]{
\includegraphics[width=3.5in,height=2.75in]{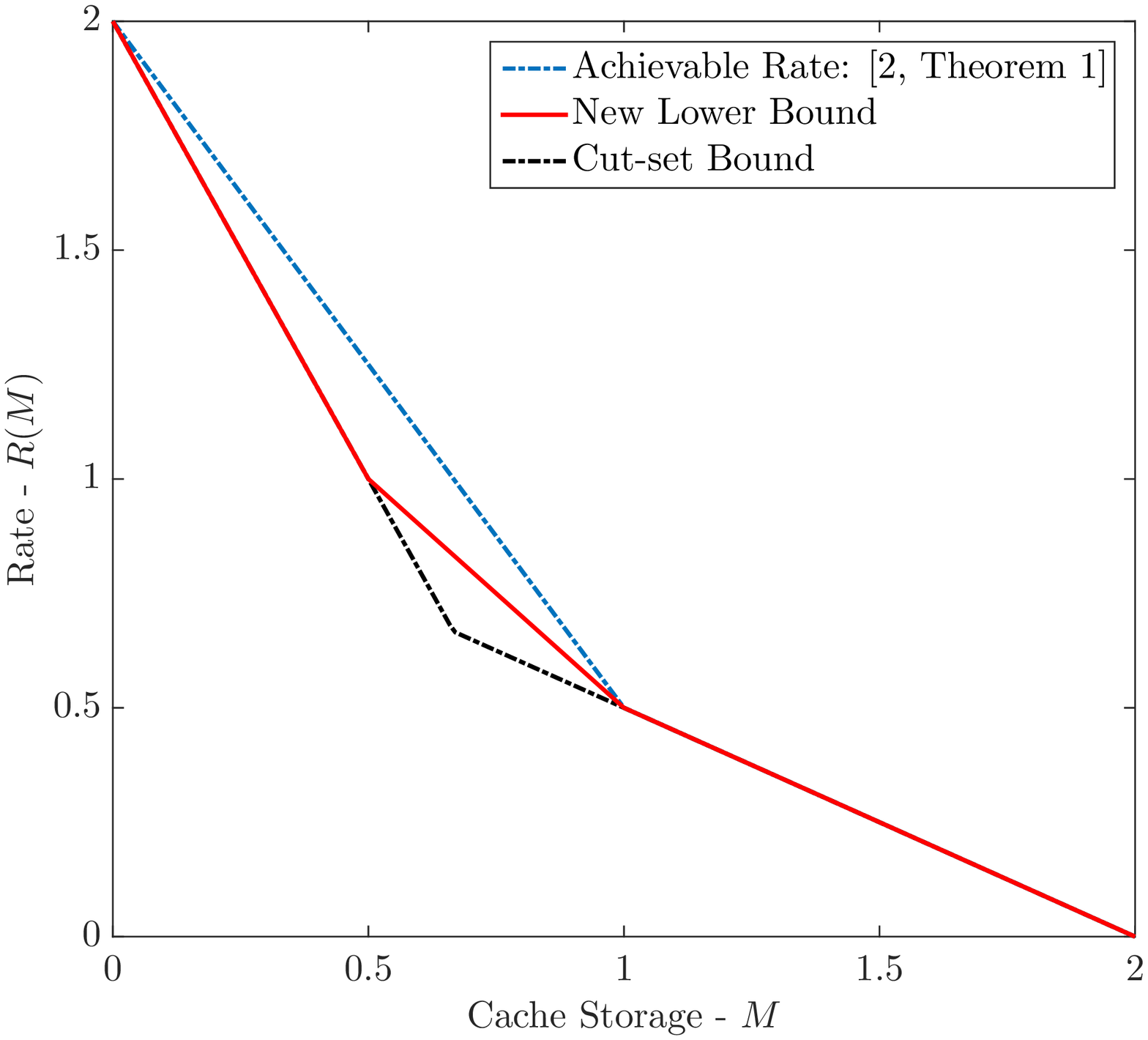}
\label{fig:nk22}
}\vspace{-10pt}
\caption{Storage-rate tradeoff for centralized content delivery with $L=1$ for $(a)$ $N=K=3$ and $(b)$ $N=K=2$. } \vspace{-10pt}
\end{figure*}
Now considering all possible subsets of $Z_{[1:3]}$ with cardinality $2$, in the RHS of (\ref{eq:ex2:2}), we have:
\begin{align}
3B & \leq 2MB + R_{\mathsf{cen}}^*B +  H(Z_{[1:3]}|A,B) - H(Z_{[2,3]}|A,B)\label{eq:ex2:3}\end{align}\begin{align}
3B & \leq 2MB + R_{\mathsf{cen}}^*B +  H(Z_{[1:3]}|A,B) - H(Z_{[1,3]}|A,B).\label{eq:ex2:4}
\end{align}
Summing (\ref{eq:ex2:2})-(\ref{eq:ex2:4}), and normalizing by $3$, we have: 

\begin{align}
\hspace{-5pt} 3B & \leq 2MB + R_{\mathsf{cen}}^*B +  H(Z_{[1:3]}|A,B) - \hspace{-4pt}\sum_{\substack{i,j=1,~ i \neq j}}^{3}\frac{H(Z_{[i,j]}|A,B)}{3}. \label{eq:ex2:5}\hspace{-2pt}
\end{align}

We next state Han's Inequality \cite[Theorem 17.6.1]{cover} on subsets of random variables, which we use for further upper bounding (\ref{eq:ex2:5}) in order to derive the proposed lower bound. 

 \textbf{Han's Inequality}: Let $Y_{[1:n]}$ denote a set of random variables. Further, let $\left(\mathsf{Y}_{[m]},\mathsf{Y}_{[r]} \right) \subseteq Y_{[1:n]}$ denote subsets of cardinality $m,r$ with $m\leq r$. Han's Inequality states that
\begin{align}\label{eHans}
\frac{1}{{n \choose r}}\sum_{\mathsf{Y}_{[r]}:\left|\mathsf{Y}_{[r]}\right|=r}\frac{H\left(\mathsf{Y}_{[r]}\right) }{r} \leq \frac{1}{{n \choose m}}\sum_{\mathsf{Y}_{[m]}:\left|\mathsf{Y}_{[m]}\right|=m}\frac{H\left(\mathsf{Y}_{[m]}\right) }{m}, 
\end{align}
where the sums are over all subsets of cardinality $r,m$ respectively. Next, from (\ref{eq:ex2:5}), consider the set of random variables $Z_{[1:3]}$ and its subsets $\left(Z_{[1,2]},Z_{[1,3]},Z_{[2,3]}\right)$ of cardinality $2$. Applying Han's Inequality for these random variables, using $n=r=3$ and $m=2$ in \eqref{eHans}, we have:
\begin{align}
\frac{2H\left( Z_{[1:3]} |A,B\right)}{3} \leq  \sum_{\substack{i,j=1, i \neq j}}^{3} \frac{ H\left(Z_{[i,j]}|A,B\right)}{3}\label{eq:ex2:6}.
\end{align}
Substituting (\ref{eq:ex2:6}) into (\ref{eq:ex2:5}), we have:
\begin{align}\label{eq:nlb1proof}
3B & \leq 2MB + R_{\mathsf{cen}}^*B +  H(Z_{[1:3]}|A,B) - \frac{2}{3}H(Z_{[1:3]}|A,B)\nonumber\\
	 & \leq 2MB + R_{\mathsf{cen}}^*B +  \frac{1}{3}H(Z_{[1:3]}|A,B)  \leq 2MB + R_{\mathsf{cen}}^*B +  \frac{1}{3}H(Z_{[1:3]},C|A,B)\nonumber\\
   & \leq 2MB + R_{\mathsf{cen}}^*B +  \frac{1}{3} \left(\underbrace{H(C|A,B)}_{\leq 1} + \underbrace{H(Z_{[1:3]}|A,B,C)}_{= 0}\right) \leq 2MB + R_{\mathsf{cen}}^*B +  \frac{1}{3}B.
\end{align}
Rearranging \eqref{eq:nlb1proof}, we get the new lower bound given by the first inequality in (\ref{nlb2}). The second bound in \eqref{nlb2} can be obtained similarly by considering $s=1$ cache and bounding the entropy of $\ell=2$ transmissions by the optimal rate $R^*_{\mathsf{cen}}$ and following steps similar to \eqref{eq:ex2:1}-\eqref{eq:nlb1proof}.

\begin{remark}
We note that the key distinction from the cut-set bounds is the mechanism of bounding the joint entropy of random variables representing the multicast transmissions and the stored contents. Specifically, considering the first inequality in \eqref{eq:ex2:1}, a naive upper bound on the term $H(X_{BCA}|Z_{[1,2]},X_{ABC})$ would be $R_{\mathsf{cen}}^{*}$, which would lead to $3\leq 2M+2R_{\mathsf{cen}}^*$, which is a loose bound. The main idea is to first observe that given $Z_{[1,2]}$ and the multicast transmission $X_{ABC}$, the files $A, B$ can be recovered. Hence, we expect a dependence between $X_{BCA}$ and the random variables in the conditioning. In order to capture this dependency, we consider multiple such requests over time, allowing us to write \eqref{eq:ex2:3}, and \eqref{eq:ex2:4}, similar to \eqref{eq:ex2:2}. This symmetrization argument directly leads to the use of Han's inequality and subsequently to the new lower bound. This is the key approach behind Corollary \ref{cor:convfund} and Theorem~\ref{th:ldem} which is a general result and holds for all problem parameters. 
\end{remark}

\begin{remark}
Recently \cite{improve_fund,gunduz_ach,piantanida_ach} proposed caching and delivery schemes which improve upon the original multicasting scheme presented in \cite[Theorem 1]{Maddah-Ali}. Specifically, \cite{improve_fund} showed that for $K\geq N$, in the small buffer region of $M = 1/K$, the achievable rate is given by $N(1-M)$ which improves on the achievable rate in \cite[Theorem 1]{Maddah-Ali}. For $N=K=3$, the new achievable point $(M,R) = (1/3,2)$ is highlighted in Fig. \ref{fig:nk33}. The lower bound in \cite[Th, 2]{Maddah-Ali} is shown to be tight only in the regime $0\leq M\leq 1/K$ for $K\geq N$ in \cite{improve_fund}. The lower bound presented in Corollary \ref{cor:convfund} shows that this is indeed the case and that the new converse is tighter than the cut-set based lower bound for $M>1/K$ as shown in in Fig. \ref{fig:nk33}.
\end{remark}

\begin{remark}
In \cite{Maddah-Ali}, the authors characterize the optimal storage-rate tradeoff for the case of $N=K=2$ and show that their lower bound, given by $R_{\mathsf{cen}}^* + 2M \geq 2 $ and $2R_{\mathsf{cen}}^* + M \geq 2$, is indeed loose. Our proposed lower bound yields the additional bound, $2R_{\mathsf{cen}}^* + 2M \geq 3$, which makes it tighter than the cut set bound. From Fig. \ref{fig:nk22} and \cite{Maddah-Ali} it can be seen that the proposed converse characterizes the optimal rate for the case of $N=K=2$. 
\end{remark}

\subsection{D2D-assisted Content Delivery: Intuition Behind Proof of Theorem \ref{th:d2d_ldem}}\label{sec:intui}
We next follow up the discussion in the previous section with an additional example to highlight our proposed techniques for the case of D2D-assisted content delivery with $L=1$ demand per user. To this end, consider again a system with $N=3$ files $(A,B,C)$ and $K=3$ users, each with a cache storage of $M\geq 1$. The proposed lower bound in Corollary \ref{cor:convd2d} gives following bounds for different $s,\ell$:
\begin{align}
&\text{New Lower Bounds:}\hspace{-75pt}& ~R_{\mathsf{d2d}}^* + 6M  &\geq 8,   ~~~~ s = 2,~\ell = 1\label{dnlb1}\\
&&8R_{\mathsf{d2d}}^* + 6M  &\geq 15,  ~~~  s = 1,~\ell = 2\label{dnlb2}\\
&\text{Cut-set Bound:}\hspace{-75pt}&2R_{\mathsf{d2d}}^* + ~M  &\geq 3,   ~~~~ s = 1,~\ell = 3,\label{cb1} 
\end{align}
where (\ref{cb1}), along with the looser bound $R_{\mathsf{d2d}}^* + 3M \geq 3$, recovers the cut set bound in \cite[Theorem 2]{Ji_Ldem}. To facilitate the derivation of the new bounds, we first consider the request vectors $(d_1,d_2,d_3) = (A,B,C)$ and $(d_1,d_2,d_3) = (B,C,A)$. The first $s=2$ cache contents $Z_{[1,2]}$ along with two composite transmissions $X_{ABC}= \{X^3_{ABC}\}, X_{BCA}= \{X^3_{BCA}\}$ from the \textit{third} user device are able to decode all $3$ files. Here each transmission has the rate of $R_{\mathsf{d2d}}^*/3$. We upper bound the entropy of $\ell=1$ transmission with this rate and use the other transmission's decoding capability, in conjunction with the cache contents $Z_{[1,2]}$, to derive a tighter bound as follows.
\begin{align}
3B & \leq H(Z_{[1,2]}, X_{ABC},X_{BCA}) \leq H(Z_{[1,2]}) + H(X_{ABC},X_{BCA}|Z_{[1,2]})\nonumber	\\
	 & \leq 2MB + H(X_{ABC}) + H(X_{BCA}|Z_{[1,2]},X_{ABC})\nonumber\\
	 & \leq 2MB + \frac{R_{\mathsf{d2d}}^*}{3}B + H(X_{BCA}|Z_{[1,2]},X_{ABC},A,B)\nonumber\\
	 & \leq 2MB + \frac{R_{\mathsf{d2d}}^*}{3}B + H(X_{BCA},Z_3|Z_{[1,2]},X_{ABC},A,B)\nonumber\\
	 & \leq 2MB + \frac{R_{\mathsf{d2d}}^*}{3}B + H(Z_3|Z_{[1,2]},X_{ABC},A,B) +  H(X_{BCA}|Z_{[1:3]},X_{ABC},A,B)\nonumber\\
	 & \myleq{(a)} 2MB + \frac{R_{\mathsf{d2d}}^*}{3}B + H(Z_3|Z_{[1,2]},A,B) \label{exd:1},
\end{align}
where step \textsf{(a)} follows from the fact that $H(X_{BCA}|Z_{[1:3]},X_{ABC},A,B,C)=0$ since $X_{BCA}$ is a function of the \textit{cache contents} $Z_{[1:3]}$. Considering the term $H(Z_3|Z_{[1,2]},A,B)$, we have:
\begin{align}\label{eq:symd0}
H(Z_3|Z_{[1,2]},A,B) = H(Z_{[1:3]}|A,B) - H(Z_{[1,2]}|A,B).
\end{align}
Using (\ref{eq:symd0}) in (\ref{exd:1}), we have:
\begin{align}
3B & \leq 2MB + \frac{R_{\mathsf{d2d}}^*}{3}B +  H(Z_{[1:3]}|A,B) - H(Z_{[1,2]}|A,B).\label{exd:2}
\end{align}
Again, considering all possible subsets of $Z_{[1:3]}$ having cardinality $2$, in the RHS of (\ref{exd:2}), we have 
\begin{align}
3B & \leq 2MB + \frac{R_{\mathsf{d2d}}^*}{3}B +  H(Z_{[1:3]}|A,B) - H(Z_{[2,3]}|A,B).\label{exd:3}\\
3B & \leq 2MB + \frac{R_{\mathsf{d2d}}^*}{3}B +  H(Z_{[1:3]}|A,B) - H(Z_{[1,3]}|A,B).\label{exd:4}
\end{align}
Symmetrizing over the inequalities in (\ref{exd:2})-(\ref{exd:4}), we have:
\begin{align}
3B & \leq 2MB + \frac{R_{\mathsf{d2d}}^*}{3}B +  H(Z_{[1:3]}|A,B) - \sum_{\substack{i,j=1,i \neq j}}^3\frac{ H(Z_{[i,j]}|A,B)}{3}. \label{exd:5}
\end{align}
Next, considering the set of caches $Z_{[1:3]}$ and its subsets $Z_{[1,2]},Z_{[1,3]}Z_{[2,3]}$ of cardinality $2$ and applying Han's Inequality (as in \eqref{eHans}), we have from \eqref{exd:2}
\begin{align} \label{eq:hans_ineq40}
3B & \leq 2MB + \frac{R_{\mathsf{d2d}}^*}{3}B +  H(Z_{[1:3]}|A,B) - \frac{2H(Z_{[1:3]}|A,B)}{3}\nonumber\\
   & \leq 2MB + \frac{R_{\mathsf{d2d}}^*}{3}B +  \frac{H(Z_{[1:3]},C|A,B)}{3} ~\leq~ 2MB + \frac{R_{\mathsf{d2d}}^*}{3}B +  \frac{1}{3}B.
\end{align}
Rearranging \eqref{eq:hans_ineq40}, we get the new lower bound in \eqref{dnlb1}. Next, we consider $s=1$ device cache, $Z_1$, and three request vectors  $(d_1,d_2,d_3) = (A,B,C)$, $(d_1,d_2,d_3) = (B,C,A)$ and $(d_1,d_2,d_3) = (C,A,B)$ along with the multicast transmissions $X_{ABC}= \{X^2_{ABC},X^3_{ABC}  \}, X_{BCA} = \{X^2_{BCA},X^3_{BCA}\}, X_{CAB}= \{X^2_{CAB},X^3_{CAB}\}$ from users $2,3$, which are capable of decoding all $3$ files. In this case, each composite transmission is of rate $2R_{\mathsf{d2d}}^*/3$. We upper bound the entropy of $\ell=2$ transmissions with their rate and following similar steps as the previous case leads us to the lower bound in \eqref{dnlb2}.
Finally, considering again, $s=1$ device storage content, $Z_1$, and three request vectors  $(d_1,d_2,d_3) = (A,B,C)$, $(d_1,d_2,d_3) = (B,C,A)$ and $(d_1,d_2,d_3) = (C,A,B)$ along with three transmissions $X_{ABC}= \{X^2_{ABC},X^3_{ABC}  \}, X_{BCA} = \{X^2_{BCA},X^3_{BCA}\}, X_{CAB}= \{X^2_{CAB},X^3_{CAB}\}$ which are capable of decoding all $3$ files. Each transmission has rate $2R_{\mathsf{d2d}}^*/3$. We upper bound the entropy of $\ell = 3$ transmissions by their rates thereby recovering the cut set bound in \eqref{cb1}.
The new converse is strictly tighter than the cut set bounds. Furthermore, the proposed converse is tight at the point $M = {N}/{K} = 1$. Setting $M=1$ in \eqref{dnlb1} and comparing with the upper bound from \cite[Theorem 1]{fund_ji} yields $R^*_{\mathsf{d2d}}(1)= 2$ i.e., the achievable scheme proposed in \cite{fund_ji} is optimal at $M=1$.

\section{Comparisons with Independent Parallel Results}\label{sec:comp}
We acknowledge the  recent independent contributions from \cite{tuninetti_optimal,tifr,ghasemi,tian, Gastpar_newconv} on developing converse results for cache-aided systems. The authors in \cite{tuninetti_optimal} derive a new converse bound based on index coding for the case of centralized content delivery with $L=1$, which shows that the achievable scheme in \cite{Maddah-Ali} is optimal if \textit{uncoded} cache placement is assumed. Again, for centralized content delivery, the authors in \cite{tifr,ghasemi} also obtain improvements over the cut-set bound, for the case when $L=1$, through different approaches than ours. While a direct comparison is analytically intractable, especially owing to the algorithm based approach of \cite{ghasemi}, we present some numerical comparisons to show that our bounds supersede these bounds in certain regimes of cache storage $M$ for the single demand case. To this end, in Fig. \ref{fig:n12k6} and \ref{fig:n6k12}, we plot the result in \cite{tifr} which yields the same bound as in \cite{ghasemi} for these instances. It can be seen that our bounds are better for the case of low cache memory for both cases. Furthermore, we note that a \textit{holistic lower bound} for centralized content delivery with $L=1$ is obtained only by combination of all lower bounding approaches in literature and maximizing over the bounds yielded by each method.

The authors in \cite{tifr} do not derive a constant gap result, however, the authors in \cite{ghasemi} show a constant gap of $4$ to the achievable rate in \cite[Theorem 1]{Maddah-Ali}. We emphasize here that the analyses to obtain multiplicative gaps (as in Theorems \ref{th:ldem_gap} and \ref{th:d2d_ldem_gap}) are essentially approximations. Thus, deriving lower bounds geared towards tightening this analysis does not guarantee the best known bounds. To this end, we consider the lower bounds presented in \cite{Gastpar_newconv}. The proposed lower bounds are generally always looser than the cut-set bounds for the case of centralized content delivery with $L=1$ and by extension than the bounds presented in this paper as shown in Fig. \ref{fig:comp}. However, the authors leverage the structure of the bounds to approximate the storage-rate tradeoff to within a constant multiplicative factor of $4.7$. We note here that the analysis presented in this paper is solely for the purpose of proving the sub-optimality of cut-set bounds in a more general problem setting, i.e., $L\geq 1$, and that the gap to the optimal can be numerically tightened to $3.5$ for centralized delivery with $L=1$, which shows that the bounds are similar to those in \cite{ghasemi,Gastpar_newconv} in terms of approximately characterizing the optimal storage-rate tradeoff.
\begin{figure*}[!t]
\centering 
\subfigure[]{
\includegraphics[width=2.475in,height=2in]{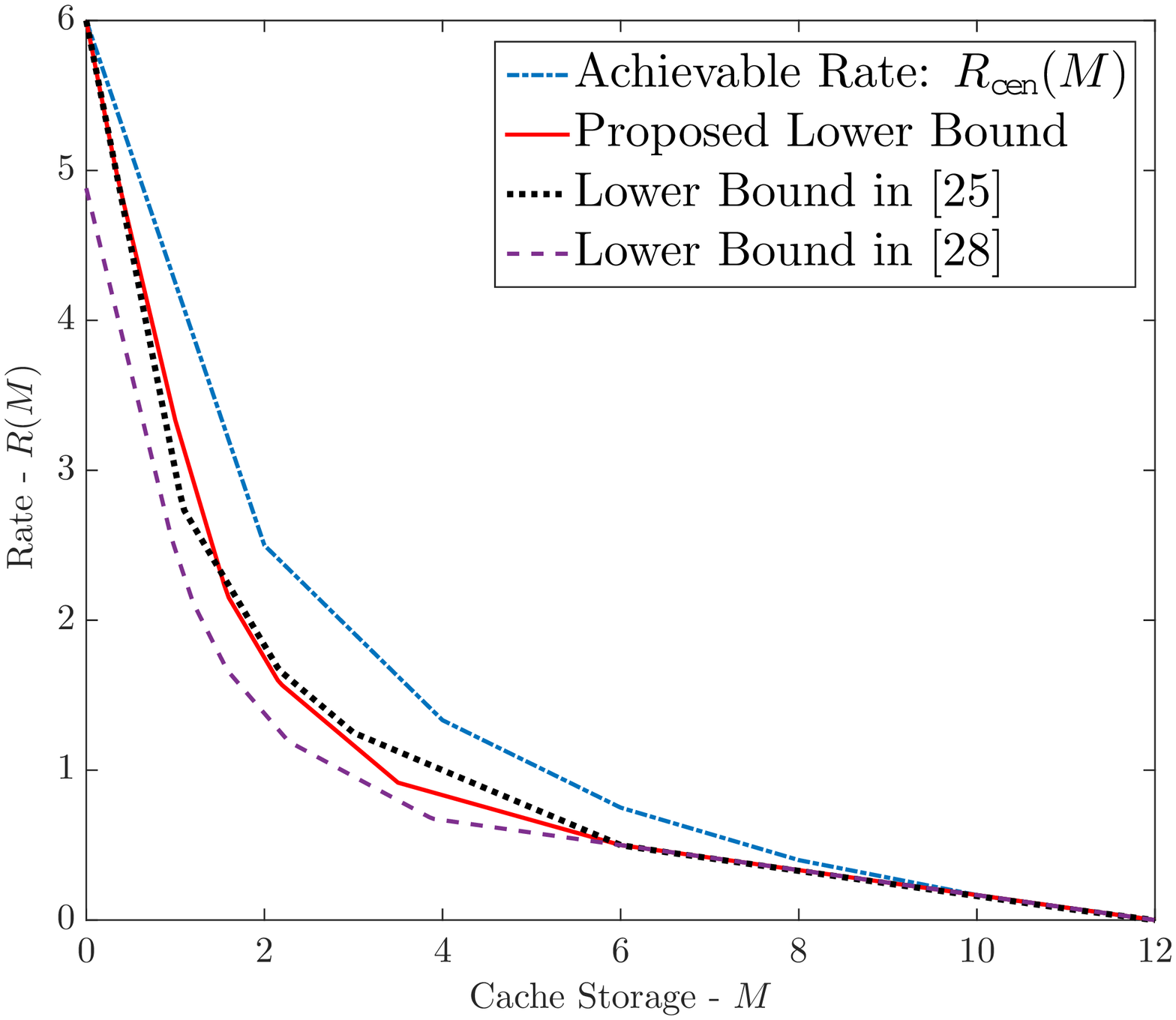}
\label{fig:n12k6}
}\hspace{-28pt}
\subfigure[]{
\includegraphics[width=2.475in,height=2in]{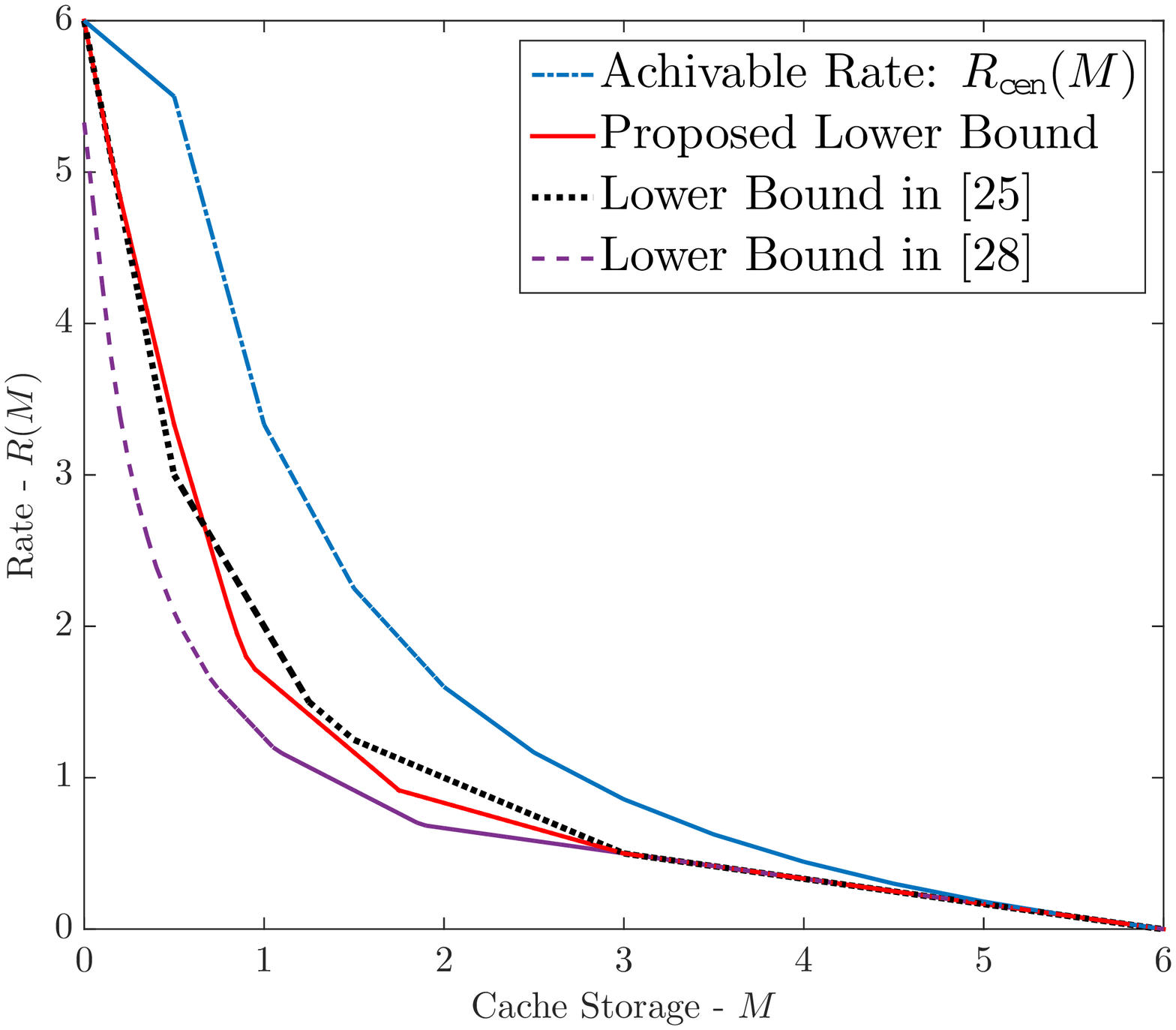}
\label{fig:n6k12}
}\hspace{-28pt}
\subfigure[]{
\includegraphics[width=2.475in,height=2in]{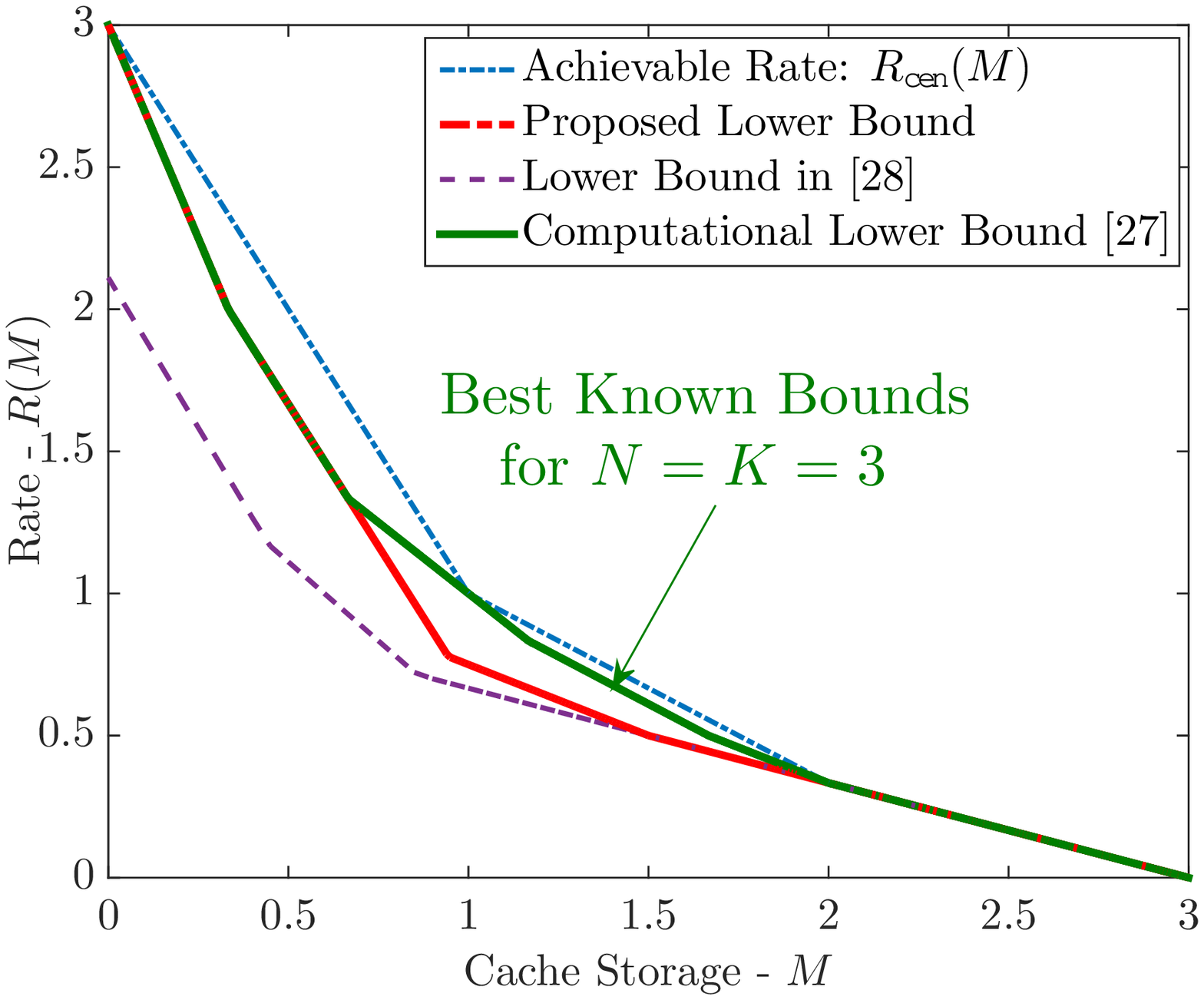}
\label{fig:tian}
}
\vspace{-10pt}
\caption{Comparisons with parallel results for the case of centralized content delivery with $L=1$ for a cache-aided system with ${(a)}$ $N=12,K=6$; ${(b)}$ $N=6,K=12$ and $(c)$ $N=K=3$.} \vspace{-10pt}
\label{fig:comp}
\end{figure*}
%

Finally, Tian \cite{tian} has recently obtained improvements for the specific case of $N=K=3$ for centralized content delivery with $L=1$, using a novel computer aided approach as shown in Fig. \ref{fig:tian}. Our proposed method recovers the bound $6M + 3R_{\mathsf{cen}}^* \geq 8$, while the approach in \cite{tifr,ghasemi} recovers the bound $M + R_{\mathsf{cen}}^* \geq 2$. However, it is unclear whether the bounds $12M + 18R^*_{\mathsf{cen}} \geq 29$ and $3M + 6R^*_{\mathsf{cen}}\geq 8$ can be tractably obtained via analytical methods. Therefore, obtaining the numerical bounds for the $N=K=3$ system with centralized delivery remains an open problem.
\section{Conclusion}\label{sec:conc}
In this paper, we presented a new technique for deriving information theoretic lower bounds for cache-aided systems with centralized as well as D2D-assisted content delivery for the general case when users can demand multiple files at each transmission interval. We leveraged Han's Inequality to better model the interaction of user caches and file decoding capabilities of multicast transmissions to derive lower bounds which are strictly tighter than existing cut-set based bounds. Leveraging the proposed lower bounds, we showed that, for the case of multiple demands per user, repeating multiple times, the schemes designed to address content delivery for single demands, is in fact order-optimal for both delivery settings. Furthermore, we provided an approximate characterization of the fundamental storage-rate tradeoff for centralized content delivery to within a constant multiplicative factor of $11$ and for D2D-assisted content delivery to within a factor of $10$ for all possible values of problem parameters, thereby improving on the existing results in both paradigms.

\appendices
\section{Proof of Theorem \ref{th:ldem}}\label{ap:ldem}
Consider a cache-aided system with $N$ files, each of size $B$ bits, and $K$ users, each with a cache size of $M$ files. Let $s$ be an integer such that $s \in \left[1:\min\{\lceil N/L\rceil,K\}\right]$. For the case of centralized delivery with $L\in [1:N]$ demands per user, the demand vector is such that each user demands $L$ distinct files at each transmission interval. Consider the first $s$ caches $Z_{[1:s]}$ and a demand vector
\begin{align}
\mathbf{D}_1 &= \Bigg(\underbrace{\mathbf{d}_{[1:s]},\mathbf{d}_{[s+1:K]}) =  \Big( [1:L],[L+1:2L],\ldots,[L(s-1)+1:Ls]}_{=~\mathbf{d}_{[1:s]}},\phi\Bigg),
\end{align}
where the first $s$ user demands are for $Ls$ unique files and last $K-s$ users' demands can be for any arbitrary $L(K-s)$ files. To service this set of demands, the central server makes a multicast transmission $X_1$, which along with the $Z_{[1:s]}$ is capable of decoding the files $F_{[1:Ls]}$. Similarly, consider another demand, 
\begin{align}
\mathbf{D}_2 &= \Big([Ls+1:L(s+1)],[L(s+1)+1:L(s+2)],\ldots,[L(2s-1):2Ls],\phi\Big),
\end{align}
and a resultant multicast transmission $X_2$, which along with the $s$ caches, are capable of decoding the files $F_{[Ls+1 : 2Ls]}$. Thus considering the demand vectors $\mathbf{D}_1,\mathbf{D}_2,\ldots,\mathbf{D}_{\lceil N/(Ls)\rceil}$ and their corresponding multicast transmissions $X_1,X_2,\ldots,X_{\lceil N/(Ls) \rceil}$, along with the first $s$ caches $Z_{[1:s]}$, the whole library of files $F_{[1:N]}$ can be decoded. Considering $B=1$ without loss of generality. We have:
\begin{align}\label{eq:Lth1:7}
N &\leq I\left(F_{1:N};Z_{[1:s]}, {X}_{[1:{\lceil N/(Ls) \rceil}]}\right) \leq H\left( Z_{[1:s]}, {X}_{[1:{\lceil N/(Ls) \rceil}]}  \right) \leq H\left( Z_{[1:s]} \right) + H\left( {X}_{[1:{\lceil{N/(Ls)}\rceil}]} | Z_{[1:s]}\right)\nonumber\\
	& \leq sM + H\left(  {X}_{[1:{\lceil{N/(Ls)}\rceil}]} | Z_{[1:s]} \right) \leq sM + H\left( {X}_{[1:{\ell}]}| Z_{[1:s]}\right) + H\left( {X}_{[\ell + 1 : {\lceil{N/(Ls)}\rceil}]}| Z_{[1:s]}, {X}_{[1:{\ell}]} \right)\nonumber\\
	& \myleq{(a)} sM + \ell R_{\mathsf{cen,L}}^*(M) +  H\left({X}_{\left[\ell + 1 : {\lceil{N/(Ls)}\rceil}\right]}| Z_{[1:s]}, {X}_{[1:{\ell}]}, F_{[1:{L\ell s}]} \right)\nonumber\\
	& \myleq{(b)} sM + \ell R_{\mathsf{cen,L}}^*(M) + H\big( {X}_{[\ell + 1: \lceil{N/(Ls)}\rceil]},Z_{[s+1:s+\mu]}| Z_{[1:s]}, {X}_{[1:{\ell}]}, F_{[1:{L\ell s}]} \big)\nonumber\\
	& \leq sM + \ell R_{\mathsf{cen,L}}^*(M) + \underbrace{H\left(  Z_{[{s+1}:{s+\mu}]}| Z_{[1:s]}, {X}_{[1:{\ell}]}, F_{[1:{L \ell s}]} \right)}_{\triangleq \delta} \nonumber\\
	&~~~~~~~~~~~~~~~~~~~~~~~~~~~~~~~~~~~~~~~~~~~~~~~+ \underbrace{H\left( {X}_{[{\ell + 1}:{\lceil{N/(Ls)}\rceil}]}|  Z_{[1:{s+\mu}]}, {X}_{[1:{\ell}]}, F_{[1:{L\ell s}]} \right)}_{\triangleq\lambda},
\end{align}
where step \textsf{(a)} results from bounding the entropy of $\ell \in \{1,2,\ldots,\lceil{N/(Ls)}\rceil\}$ transmissions given the caches $Z_{[1:s]}$ by $\ell R_{\mathsf{cen,L}}^*(M)$, where each transmission is of rate $R_{\mathsf{cen,L}}^*(M)$. Furthermore, the caches $Z_{[1:s]}$ with transmissions ${X}_{[1:\ell]}$ can decode files $F_{[1: L\ell s]}$. In step \textsf{(b)}, $\mu$ number of caches are introduced into the entropy, where $\mu$ is the number of remaining caches which along with caches $Z_{[1:s]}$ and transmissions ${X}_{[1:\ell]}$, can decode the remaining $(N - L \ell s)$ files. It is to be noted that all the remaining $K-s$ caches might not be required for decoding all files. Thus we have: 
\begin{align}
\mu = \min\left\{ \left\lceil\frac{ N - L\ell s}{L\ell}\right\rceil, K-s\right\} = \min\left\{ \left\lceil {N}/{(L\ell)}\right\rceil, K\right\} - s,
\end{align}
where the last equality follows since $s$ is an integer. Next, we obtain upper bounds on the two terms $\delta$ and $\lambda$ in \eqref{eq:Lth1:7}.

\n \textbf{\textit{Upper Bound on}} $\bm{\delta}:$ We consider the factor $\delta$, from \eqref{eq:Lth1:7} and upper bound it as follows:
\begin{align}
\delta & = H\left(  Z_{[{s+1}:{s+\mu}]}| Z_{[1:s]}, {X}_{[1:{\ell}]}, F_{[1:{L\ell s}]} \right)  \leq   H\left(  Z_{[{s+1}:{s+\mu}]}| Z_{[1:s]}, F_{[1:{L\ell s}]}\right)\nonumber\\
			 & = H\left( Z_{[1:{s + \mu}]}|F_{[1:{L\ell s}]} \right) - H\left(  Z_{[1:s]}|F_{[1:{L\ell s}]} \right). \label{Lth1d:2}
\end{align}
Considering all possible subsets of $Z_{[1:{s + \mu}]}$ having cardinality $s$, i.e., considering all possible combinations of distinct files in the request vectors and all possible combinations of $s$ caches in \eqref{eq:Lth1:7}, we can obtain ${s+\mu \choose s}$ different inequalities of the form of (\ref{Lth1d:2}). Symmetrizing over all the inequalities, we have:
\begin{align}
&\delta \leq H\left( Z_{[1:{s + \mu}]} |F_{[1:{L\ell s}]}\right)  - \sum_{i=1}^{{s+\mu \choose s}} \frac{ H\left(\mathsf{Z}^i_{[s]}|F_{[1:{L\ell s}]} \right)}{{s+\mu \choose s}}\label{Lhan1},
\end{align} 
where, $\mathsf{Z}^i_{[s]}$ is the $i$-th subset of $Z_{[1:{s + \mu}]}$ with cardinality $s$. Next, consider $Z_{[1:s+\mu]}$ as the set of random variables $\{Z_{k} : k \in 1,\ldots, s+\mu\}$ and the subsets $\mathsf{Z}^i_{[s]} \subseteq Z_{[1:s + \mu]}, ~\forall i = 1,\ldots, {s+\mu \choose s}$. Applying Han's Inequality from \eqref{eHans}, we have:
\begin{align}
 \frac{s}{s+\mu}H\left( Z_{[1:s+\mu]} |F_{[1:{L\ell s}]}\right) \leq  \frac{1}{{s+\mu \choose s}}\sum_{i=1}^{{s+\mu \choose s}}  H\left(\mathsf{Z}^i_{[s]}|F_{[1:{L\ell s}]}\right)\label{Lhan2:2}.
\end{align}
Substituting (\ref{Lhan2:2}) into (\ref{Lhan1}), we have:
\begin{align} \label{eq:Lth1d:7}
\delta & \leq H\left( Z_{[1:s+\mu]} |F_{[1:{L\ell s}]}\right) - \frac{s}{s + \mu}H\left( Z_{[1:s+\mu]} |F_{[1:{L\ell s}]}\right)\nonumber\\
			 & = \frac{\mu}{s + \mu} H\left( Z_{[1:s+\mu]} |F_{[1:{L\ell s}]}\right) \leq \frac{\mu}{s + \mu} H\left(  Z_{[1:s+\mu]},F_{[{L\ell s+1}:N]}|F_{[1:{L\ell s}]} \right)\nonumber\\
			 & = \frac{\mu}{s + \mu}\left( H\left(F_{[{L\ell s+1}:N]}|F_{[1:{L\ell s}]} \right) +  \underbrace{H\left(  Z_{[1:s+\mu]}|F_{[1:N]} \right)}_{=0} \right) \myleq{(a)} ~\frac{\mu}{s + \mu}(N - L\ell s)^{+},
\end{align}
where step \textsf{(a)} follows from the fact that the caches are functions of all $N$ files in the library. 

 \n \textbf{\textit{Upper Bound on}} $\bm{\lambda}:$ To upper bound $\lambda$, we observe from the last step in \eqref{eq:Lth1:7} that the transmissions $X_{[1:\ell]}$, along with caches $Z_{[1:s+\mu]}$ can decode the files $F_{[1:L\ell(s+\mu)]}$ within the conditioning, i.e.,
\begin{align}
\lambda = H\left( {X}_{[{\ell + 1}:{\lceil{N/(Ls)}\rceil}]}|  Z_{[1:{s+\mu}]}, {X}_{[1:{\ell}]}, F_{[1:{L\ell(s + \mu)}]} \right).
\end{align}
In order to characterize the upper bound on $\lambda$, we consider two cases as follows.

\n $\bullet$ \textbf{\textit{Case}} $\bm{\mathit{1 \left(N \leq L\ell(s + \mu)\right):}}$ All files are decoded by the caches $Z_{[1:{s+\mu}]}$ and transmissions ${X}_{[1:{\ell}]}$ within the conditioning for the term $\lambda$ in \eqref{eq:Lth1:7}. We have
\begin{align}
\lambda & = H\left( {X}_{[{\ell + 1}:{\lceil{N/(Ls)}\rceil}]}|  Z_{[1:{s+\mu}]}, {X}_{[1:{\ell}]}, F_{[1:N]} \right) = 0\label{Lth1l:0}, 
\end{align}
since all transmissions are functions of the file library $F_{[1:N]}$. In the case when, for $N > K$, fewer than $K$ caches suffices to decode all files with the transmissions within the conditioning in $\lambda$ i.e. $s+\mu \leq K$, we have:
\begin{align}
KL\ell \geq L\ell(s + \mu) \geq N, ~~~\text{i.e.,}~~~ \lambda = (N - KL\ell)^+ ~~= 0.
\end{align} 
It can also be easily seen that for the case of $K\geq N$, $\lambda = (N - KL\ell)^+ = 0 $ since $\ell,L \geq 1$.

\n $\bullet$ \textbf{\textit{Case}} $\bm{\mathit{2 \left(N > L\ell(s + \mu)\right):}}$ The case when, even with $s+\mu = K$ caches, all files are not decoded by the caches and transmissions within the conditioning for the term $\lambda$ in \eqref{eq:Lth1:7}. In this case, $\lambda \neq 0$ and we have:
\begin{align} \label{Lth1l:4}
\lambda &= H\left( {X}_{[{\ell + 1}:{\lceil{N/(Ls)}\rceil}]}|  Z_{[1:{s+\mu}]}, {X}_{[1:{\ell}]}, F_{[1:{KL\ell}]} \right)\nonumber\\
				&\leq H\left({X}_{[{\ell + 1}:{\lceil{N/(Ls)}\rceil}]}, F_{[{KL\ell + 1}:{N}]}|  Z_{[1:{s+\mu}]}, {X}_{[1:{\ell}]}, F_{[1:{KL\ell }]} \right)\nonumber\\
        &\leq H\left(F_{[{KL\ell + 1}:{N}]}|  Z_{[1:{s+\mu}]}, {X}_{[1:{\ell}]}, F_{[1:{KL\ell }]}\right) + H\left({X}_{[{\ell + 1}:{\lceil{N/(Ls)}\rceil}]}|   Z_{[1:{s+\mu}]}, {X}_{[1:{\ell}]}, F_{[1:N]} \right)\nonumber\\
				& \myleq{(a)}  H\left(F_{[{KL\ell + 1}:{N}]}\right) \leq (N - KL\ell),
\end{align}
where step \textsf{(a)} follows from the fact that the second entropy term in the previous step goes to zero since transmissions are functions of the $N$ files. Thus from (\ref{Lth1l:0}) and (\ref{Lth1l:4}), we can compactly bound $\lambda$ as:
\begin{align}\label{eq:Llambda}
\lambda \leq (N - KL\ell)^+. 
\end{align}
Substituting (\ref{eq:Lth1d:7}) and (\ref{eq:Llambda}) into (\ref{eq:Lth1:7}), we have:
\begin{align} \label{eq:lbcen_final}
N ~~~& \leq sM + \ell R_{\mathsf{cen,L}}^*(M) + \frac{\mu }{s + \mu}(N - L\ell s)^{+} +  (N - KL\ell)^+ 
\end{align}
Rearranging \eqref{eq:lbcen_final}, we obtain the following lower bound on the optimal rate $R_{\mathsf{cen,L}}^*(M)$
\begin{align}
R_{\mathsf{cen,L}}^*(M) &\geq  \frac{1}{\ell}\left\{ N - sM - \frac{\mu }{s+\mu}(N - L\ell s )^{+} - (N - KL\ell )^{+}\right\}.
\end{align}
Optimizing over all parameter values of $s,\ell$, 
completes the proof of Theorem \ref{th:ldem}.

\section{Proof of Theorem \ref{th:ldem_gap}}\label{ap:ldem_gap}

\newcommand{\mo}{0.3049}
\newcommand{\mt}{0.442}
\newcommand{\lo}{1.275}
\newcommand{\lt}{0.2}
\newcommand{\bo}{0.9649}
\newcommand{\bt}{0.984}
\newcommand{\con}{10}

From Theorem \ref{th:ldem}, considering the lower bound on the optimal rate $R_{\mathsf{cen,L}}^*(M)$, we set $\ell = \left\lceil \frac{\beta N}{Ls}\right\rceil \in \left[1:\left\lceil \frac{N}{Ls}\right\rceil\right]$ with $\beta\in[0,1]$. Using this, we next derive an upper bound on the term $\left(\frac{\mu}{\mu + s} \right)$ as follows
\begin{align}\label{eq:boundmu_ldem}
\frac{\mu}{\mu + s} & = \frac{\min\left\{\left\lceil \frac{N}{L\ell}\right\rceil,K\right\} - s}{\min\left\{\left\lceil \frac{N}{L\ell}\right\rceil,K\right\}} \leq 1 - \frac{s}{\left\lceil \frac{N}{L\ell}\right\rceil} =  1 - \frac{s}{\left\lceil \frac{N}{L\left\lceil \frac{\beta N}{Ls}\right\rceil}\right\rceil}\leq 1 - \frac{s}{\left\lceil \frac{s}{\beta} \right\rceil} \nonumber\\
										&\leq 1 - \frac{s}{\frac{s}{\beta} + 1} = 1 - \frac{\beta}{1 + \frac{\beta}{s}} \leq 1 - \frac{\beta}{1 + \beta} = \frac{1}{1+\beta},
\end{align}
where the last inequality follows from the fact that $s\geq 1$. Substituting (\ref{eq:boundmu_ldem}) into (\ref{eq:ldem}), we have:
\begin{align}\label{eq:ldem_lb}
R_{\mathsf{cen,L}}^*(M) & \geq \frac{N - sM - \frac{1}{1 + \beta}\left(N - L\left\lceil\frac{\beta N}{Ls}\right\rceil s\right)^+ - \left(N - KL\left\lceil\frac{\beta N}{Ls}\right\rceil\right)^+}{\left\lceil\frac{\beta N}{Ls}\right\rceil}\nonumber\\
				  &\geq \frac{\left(\frac{2\beta}{1 + \beta}\right)N - sM -  N\left(1 - K\frac{\beta}{s}\right)^+}{\left\lceil\frac{\beta N}{Ls}\right\rceil}.
\end{align}

\n Next, we consider two cases, namely $(i)$ $\min\left\{\frac{N}{L},K\right\}\leq \con$; and $(ii)$ $\min\left\{\frac{N}{L},K\right\}\geq 11$. 

\n $\bullet$ \textbf{\textit{Case}} $\bm{\mathit{1 \left(\min\left\{\frac{N}{L},K\right\}\leq \con\right):}}$
For this case, setting $s=1$ and $\beta=1$ in \eqref{eq:ldem_lb}, we have the following form on the lower bound,
\begin{align}
R^*_{\mathsf{cen,L}}(M)\geq \frac{N\left(1 - \frac{M}{N}\right)}{\left\lceil \frac{N}{L} \right\rceil}
\end{align}
Consider first, the case when $\frac{N}{L}\leq K$. From \eqref{eq:ldem_ach}, we have the following upper bound on the achievable rate 
\begin{align}
R_{\mathsf{cen,L}}(M)\leq \min\{N,KL\}\left(1 - \frac{M}{N}\right) \leq N\left(1 - \frac{M}{N}\right).
\end{align}
Therefore, we have 
\begin{align}\label{eq:ldgap_1}
\text{Gap} = \frac{R_{\mathsf{cen,L}}(M)}{R^*_{\mathsf{cen,L}}(M)} \leq \left\lceil \frac{N}{L} \right\rceil \leq 10.
\end{align}
Next, consider the case when $K\leq \frac{N}{L}$. Again, from \eqref{eq:ldem_ach}, we have the following upper bound on the achievable rate
\begin{align}
R_{\mathsf{cen,L}}(M)\leq \min\{N,KL\}\left(1 - \frac{M}{N}\right) \leq KL\left(1 - \frac{M}{N}\right).
\end{align}
Again, setting $s=1$ and $\beta=1$ in \eqref{eq:ldem_lb}, we have
\begin{align}
R^*_{\mathsf{cen,L}}(M)\geq \frac{N\left(1 - \frac{M}{N}\right)}{\frac{N}{L}+1} = \frac{L\left(1 - \frac{M}{N}\right)}{1+\frac{L}{N}} \geq \frac{KL\left(1 - \frac{M}{N}\right)}{1+K}
\end{align}
Therefore, we have
\begin{align}\label{eq:ldgap_2}
\text{Gap} = \frac{R_{\mathsf{cen,L}}(M)}{R^*_{\mathsf{cen,L}}(M)} \leq K+1  \leq \con+1 = 11.
\end{align}

\n $\bullet$ \textit{\textbf{Case}} $\bm{\mathit{2 \left(\min\left\{\frac{N}{L},K\right\}\geq 11\right):}}$
For this case, we consider three distinct regimes for the cache storage size $M$: \textit{Regime 1:} ~ $0 \leq M \leq \lo \max\left\{L,{N}/{K}\right\}$; {\textit{Regime 2:}}~ $\lo \max\left\{L,{N}/{K}\right\} < M \leq \lt N$; and {\textit{Regime 3:}} ~ $\lt N < M \leq N$. We consider each of the three regimes separately.

\begin{itemize}
\item \textbf{\textit{Regime}}~$\bm{\mathit{1 \left(0 \leq M \leq \lo \max\left\{L,{N}/{K}\right\}\right):}}$ 

For this regime, we set $s = \lfloor \mo \min\{N/L,K\} \rfloor \in [1:\min\{N/L,K\}]$ and $\ell = \left\lceil \frac{\bo N}{Ls} \right\rceil$, from \eqref{eq:ldem_lb}, we have
\begin{align}\label{eq:reg1}
&R^*_{\mathsf{cen,L}}(M)~\geq~ \frac{\left(\frac{2\times\bo}{1 + \bo}\right) - s\frac{M}{N} -  \left(1 - K\frac{\bo}{s}\right)^+}{\frac{\bo}{Ls} + \frac{1}{N}}\nonumber\\
												&= \frac{\left(\frac{2\times\bo}{1 + \bo}\right) - \lfloor \mo \min\{N/L,K\} \rfloor\frac{\lo \max\left\{L,{N}/{K}\right\}}{N} -  \left(1 - K\frac{\bo}{\lfloor \mo \min\{N/L,K\} \rfloor}\right)^+}{\frac{\bo}{L\lfloor \mo \min\{N/L,K\} \rfloor} + \frac{1}{N}}\nonumber\\
												&\mygeq{(a)} \frac{\left(\frac{2\times\bo}{1 + \bo}\right) - (\mo\times\lo)  \frac{\min\{N/L,K\}\max\left\{L,{N}/{K}\right\}}{N} -  \left(1 - K\frac{\bo}{\mo \min\{N/L,K\}}\right)^+}{\frac{\bo}{L\left(\mo \min\{N/L,K\}-1\right)} + \frac{1}{N}}\nonumber\\
												&\geq \frac{L\min\left\{\frac{N}{L},K\right\}\left(\mo - \frac{1}{\min\left\{\frac{N}{L},K\right\}} \right) \left\{ \left(\frac{2\times\bo}{1 + \bo}\right) - (\mo\times\lo) -  \left(1 - \frac{\bo}{\mo}\right)^+\right\}}{\bo +  \frac{L\left(\mo\min\left\{\frac{N}{L},K\right\} -1\right)}{N}} \nonumber\\
												&\mygeq{(b)} \frac{\min\left\{N,KL\right\}\left(\mo - \frac{1}{\con+1} \right) \left\{ \left(\frac{2\times\bo}{1 + \bo}\right) - (\mo\times\lo) -  \left(1 - \frac{\bo}{\mo}\right)^+\right\}}{\bo + \mo} \nonumber\\
												&\geq \frac{\min\left\{N,KL\right\}}{\con},
\end{align}
where step \textsf{(a)} follows by using $\lfloor \mo \min\{N/L,K\}\rfloor \leq \mo \min\{N/L,K\}$ in the numerator and $\lfloor \mo \min\{N/L,K\}\rfloor \geq \mo \min\{N/L,K\}-1$ in the denominator; and step \textsf{(b)} follows by using $\min\{N/L,K\}\leq N/L$ in the second term in the denominator. Again, considering the upper bound in \eqref{eq:ldem_ach}, we have
\begin{align}
R_{\mathsf{cen,L}}(M)\leq \min\{N,KL\}\left(1 - \frac{M}{N}\right) \leq \min\{N,KL\}.
\end{align}
Therefore for \textit{Regime $\mathit{1}$}, we have
\begin{align}\label{eq:ldgap_3}
\text{Gap} = \frac{R_{\mathsf{cen,L}}(M)}{R^*_{\mathsf{cen,L}}(M)} \leq \con.
\end{align}

\item \textbf{\textit{Regime}}~$\bm{\mathit{2 \left(\lo \max\left\{L,{N}/{K}\right\} < M \leq \lt N\right):}}$ 

For this regime, setting $s = \left\lfloor \mt \frac{N}{M}\right\rfloor \in [1:\min\{N/L,K\}]$\footnote{The range of $s$ is validated as follows. Using the upper bound $M\leq \lt N$, we have $\mt N/M \geq \mt/\lt \geq 1$. Again using the lower bound $M \geq \lo L$, we have $\mt N/M \leq \frac{\mt}{\lo}N/L \leq N/L$. Again using $M\geq \lo N/K$, we have $\mt N/M \leq K$. } and $\ell = \left\lceil \frac{\bt N}{Ls}\right\rceil$, from \eqref{eq:ldem_lb}, we have
\begin{align}
R^*_{\mathsf{cen,L}}(M) &\geq \frac{\left(\frac{2\times\bt}{1 + \bt}\right) - s\frac{M}{N} -  \left(1 - K\frac{\bt}{s}\right)^+}{\frac{\bt}{Ls} + \frac{1}{N}} \nonumber\\
												&= \frac{\left(\frac{2\times\bt}{1 + \bt}\right) - \left\lfloor \mt \frac{N}{M}\right\rfloor\frac{M}{N} -  \left(1 - K\frac{\bt}{\left\lfloor \mt \frac{N}{M}\right\rfloor}\right)^+}{\frac{\bt}{L\left\lfloor \mt \frac{N}{M}\right\rfloor} + \frac{1}{N}}\nonumber\\
												&\mygeq{(a)} \frac{\left(\frac{2\times\bt}{1 + \bt}\right) -  \mt \frac{N}{M}\frac{M}{N} -  \left(1 - \frac{\bt}{\mt} \frac{KM}{N}\right)^+}{\frac{\bt}{L\left(\mt \frac{N}{M} - 1\right)} + \frac{1}{N}} \nonumber\\
												&\mygeq{(b)} \frac{\frac{LN}{M} \left(\mt - \frac{M}{N} \right)\left\{\left(\frac{2\times\bt}{1 + \bt}\right) -  \mt -  \left(1 - \frac{\bt}{\mt}\times \lo \right)^+\right\}}{\bt + \frac{\mt}{\lo}\frac{L}{M}}\nonumber\\
												&\mygeq{(c)} \frac{\frac{LN}{M} \left(\mt - \lt \right)\left\{\left(\frac{2\times\bt}{1 + \bt}\right) -  \mt -  \left(1 - \frac{\bt\times \lo}{\mt} \right)^+\right\}}{\bt + \frac{\mt}{\lo}} \geq \frac{LN}{10 M},
\end{align}
where step \textsf{(a)} follows again by using $\lfloor\mt N/M \rfloor \leq \mt N/M$ in the numerator and $\lfloor \mt N/M\rfloor \geq \mt N/M-1$ in the denominator; step \textsf{(b)} follows from using $KM/N \geq \lo$; and step \textsf{(c)} follows by using $M/N \leq \lt$ in the numerator and $M\geq \lo L$ in the denominator. Again considering the upper bound in \eqref{eq:ldem_ach}, we have
\begin{align}
R_{\mathsf{cen,L}}(M)\leq \frac{KL\left(1 - \frac{M}{N}\right)}{1 + \frac{KM}{N}} \leq \frac{LN}{M}\left(1 - \frac{M}{N}\right)\leq \frac{LN}{M}.
\end{align}
Therefore for \textit{Regime 2}, we have
\begin{align}\label{eq:ldgap_4}
\text{Gap} = \frac{R_{\mathsf{cen,L}}(M)}{R^*_{\mathsf{cen,L}}(M)} \leq \con.
\end{align}

\item \textbf{\textit{Regime}}~$\bm{\mathit{3 \left(\lt N < M \leq  N\right):}}$ 

In this regime, setting $s=1$ and $\ell = \left\lceil\frac{N}{L}\right\rceil$ in \eqref{eq:ldem_lb}, we note that in this case, $\mu=0$ and $(N-K\ell)^+=0$. Thus, we have
\begin{align}
R^*_{\mathsf{cen,L}}(M) & \geq \frac{N\left(1 - \frac{M}{N}\right)}{\frac{N}{L}+1} \geq \frac{\left(1 - \frac{M}{N}\right)}{\frac{1}{L}+\frac{1}{N}}.
\end{align}
From \eqref{eq:ldem_ach}, we have
\begin{align}
R_{\mathsf{cen,L}}(M)\leq \frac{KL\left(1 - \frac{M}{N}\right)}{1 + \frac{KM}{N}} \leq \frac{LN}{M}\left(1 - \frac{M}{N}\right).
\end{align}
Therefore for \textit{Regime 3}, we have
\begin{align}\label{eq:ldgap_5}
\mathrm{Gap} = \frac{R_{\mathsf{cen,L}}(M)}{R^*_{\mathsf{cen,L}}(M)} \leq \frac{LN}{M}\left(\frac{1}{L}+\frac{1}{N}\right) \leq \frac{2N}{M} \leq \frac{2}{\lt} \leq \con
\end{align}
\end{itemize}
Combining \eqref{eq:ldgap_1},\eqref{eq:ldgap_2},\eqref{eq:ldgap_3},\eqref{eq:ldgap_4} and \eqref{eq:ldgap_5}, completes the proof of Theorem \ref{th:ldem_gap}. \qed


\section{Proof of Theorem \ref{th:gapfund}}\label{ap:gapfund}
From Corollary \ref{cor:convfund}, considering the lower bound on the optimal rate $R_{\mathsf{cen}}^*(M)$, we set $\ell = \left\lceil \frac{\beta N}{s}\right\rceil \in \{1,2,\ldots,\left\lceil \frac{N}{s}\right\rceil \}$ with $0 <\beta \leq 1$. Using this we next derive an upper bound on $\left(\frac{\mu}{\mu + s} \right)$.
\begin{align}\label{eq:boundmu}
\frac{\mu}{\mu + s} & = \frac{\min\left\{\left\lceil \frac{N - \ell s}{\ell}\right\rceil, K - s\right\}}{\min\left\{\left\lceil \frac{N - \ell s}{\ell}\right\rceil, K - s\right\} + s} = \frac{\min\left\{\left\lceil \frac{N}{\ell}\right\rceil,K\right\} - s}{\min\left\{\left\lceil \frac{N}{\ell}\right\rceil,K\right\}}\nonumber\\
										& = 1 - \frac{s}{\min\left\{\left\lceil \frac{N}{\ell}\right\rceil,K\right\}} =  1 - \frac{s}{\min\left\{\left\lceil \frac{N}{\left\lceil \frac{\beta N}{s}\right\rceil}\right\rceil,K\right\}}\nonumber\\
										&\leq 1 - \frac{s}{\left\lceil \frac{s}{\beta} \right\rceil} \leq 1 - \frac{s}{\frac{s}{\beta} + 1} = 1 - \frac{\beta}{1 + \frac{\beta}{s}} \leq 1 - \frac{\beta}{1 + \beta} = \frac{1}{1+\beta},
\end{align}
where the last inequality follows from the fact that $s\geq 1$. Substituting (\ref{eq:boundmu}) into (\ref{eq:convfund}), we have:
\begin{align}\label{eq:convlb}
R_{\mathsf{cen}}^*(M) & \geq \frac{N - sM - \frac{1}{1 + \beta}\left(N - \left\lceil\frac{\beta N}{s}\right\rceil s\right)^+ - \left(N - K\left\lceil\frac{\beta N}{s}\right\rceil\right)^+}{\left\lceil\frac{\beta N}{s}\right\rceil}\nonumber\\
				  &\geq \frac{N - sM - N\left(\frac{1 - \beta}{1 + \beta}\right) - N\left(1 - K\frac{\beta}{s}\right)^+}{\left\lceil\frac{\beta N}{s}\right\rceil}.
\end{align}

\n Next, we consider two cases, namely $(ii)$ $\min\{N,K\}\leq 8$; and $(ii)$ $\min\{N,K\}\geq 9$. We next consider each case separately.

\n $\bullet$ \textbf{\textit{Case}} $\bm{\mathit{1 \left(\min\{N,K\}\leq 8\right):}}$ For this case, setting $s=1$ and $\beta = 1$ in (\ref{eq:convlb}), we have:
\begin{align}\label{eq:fund_c1_lb}
R_{\mathsf{cen}}^*(M) \geq \frac{N - M}{N} = \left( 1 - \frac{M}{N}\right).
\end{align}
Again, from \cite[Theorem 1]{Maddah-Ali}, we have:
\begin{align}\label{eq:fund_c1_ub}
R_{\mathsf{cen}}(M) \leq \min\{N,K\}\left(1 - \frac{M}{N} \right).
\end{align}
Thus for this case, the gap between the upper and lower bound is given by:
\begin{align}\label{eq:gapc1}
\mathrm{Gap} = \frac{R_{\mathsf{cen}}(M)}{R_{\mathsf{cen}}^*(M)} \leq \min\{N,K\} \leq 8.
\end{align}

\n $\bullet$ \textbf{\textit{Case}} $\bm{\mathit{2 \left(\min\{N,K\}\geq 9\right):}}$ For this case, we consider three distinct regimes for the cache storage size $M$ namely $(i)$ \textit{Regime 1:} ~ $0 \leq M \leq 1.01 \max\left\{1,{N}/{K}\right\}$; $(ii)$ {\textit{Regime 2:}}~ $1.01 \max\left\{1,{N}/{K}\right\} < M \leq 0.1250N$; and $(ii)$ {\textit{Regime 3:}} ~ $0.1250N < M \leq N$. 
We next consider each of the three regimes separately.

\begin{itemize}
\item \textbf{\textit{Regime}} $\bm{\mathit{1 \left(0 \leq M \leq 1.01 \max\left\{1,{N}/{K}\right\}\right):}}$ 

\n In this regime, setting $s = \lfloor 0.4701 \min\{N,K\}\rfloor \in \{1,2,\ldots,\min\{N,K\}\}$, $\ell = \left\lceil\frac{0.93N}{s} \right\rceil$, and using the fact that $x \leq \lceil x\rceil\leq x+1$ and $x-1 \leq \lfloor x\rfloor \leq x$, we have:
\begin{align}
R_{\mathsf{cen}}^*(M) &\geq \frac{N - sM - N\left(\frac{1 - \beta}{1+\beta}\right) - N\left(1 - \frac{K\beta}{s}\right)^+}{\frac{\beta N}{s} + 1}\geq \frac{N\left[\frac{2\beta}{1+\beta}  - s\frac{M}{N} - \left(1 - \frac{K\beta}{s}\right)^+\right]}{\frac{\beta N}{s} + 1}\nonumber\\
					&\geq \frac{\Bigg\{\frac{2\times0.93}{1+ 0.93}  - \lfloor 0.4701 \min\{N,K\}\rfloor\frac{M}{N} - \left(1 - \frac{0.93 K}{\lfloor 0.4701 \min\{N,K\}\rfloor}\right)^+ \Bigg\}}{\frac{0.93}{\lfloor 0.4701 \min\{N,K\}\rfloor} + \frac{1}{N}}  \nonumber\\
					&\geq \frac{\Bigg\{\frac{2\times0.93}{1+ 0.93}  - 0.4701 \min\{N,K\}\frac{1.01 \max\left\{1,{N}/{K}\right\}}{N} - \left(1 - \frac{0.93 K}{0.4701 \min\{N,K\}}\right)^+\Bigg\}}{\frac{0.93 }{0.4701 \min\{N,K\} - 1} + \frac{1}{N}}\nonumber\\
					&\geq \frac{\Bigg\{\left(0.4701 \min\{N,K\} - 1\right)\Bigg[\frac{2\times0.93}{1+ 0.93}  - 0.4701 \times 1.01 - \left(1 - \frac{0.93}{0.4701}\right)^+\Bigg]\Bigg\}}{0.93  + \frac{0.4701 \min\{N,K\}}{N} - \frac{1}{N}}\nonumber\\
					&\geq \min\{N,K\}\frac{\left(0.4701 - \frac{1}{9}\right)}{0.93  + 0.4701} \Bigg[\frac{2\times0.93}{1+ 0.93}  - 0.4701 \times 1.01 - \left(1 - \frac{0.93}{0.4701}\right)^+\Bigg]\nonumber\\
					&\geq \frac{\min\{N,K\}}{8}.
\end{align}

Again, from \cite[Theorem 1]{Maddah-Ali}, we have:
\begin{align}
R_{\mathsf{cen}}(M) \leq \min\{N,K\}\left(1 - \frac{M}{N} \right) \leq \min\{N,K\}.
\end{align}
Thus for this regime, the gap between the upper and lower bound is given by:
\begin{align}\label{eq:gapc2r1}
\mathrm{Gap} = \frac{R_{\mathsf{cen}}(M)}{R_{\mathsf{cen}}^*(M)} \leq 8.
\end{align}

\item \textbf{\textit{Regime}} $\bm{\mathit{2 \left(1.01 \max\left\{1,{N}/{K}\right\} < M \leq 0.1250N \right):}}$

\n In this regime, we set $s = \lfloor 0.4983 \frac{N}{M}\rfloor \in \{1,2,\ldots,\min\{N,K\}\}$, $\ell = \left\lceil\frac{0.991N}{s} \right\rceil$ and using the fact that $x \leq \lceil x\rceil\leq x+1$ and $x-1 \leq \lfloor x\rfloor \leq x$, we have:
\begin{align}
R_{\mathsf{cen}}^*(M)	& \geq \frac{N\left[\frac{2\beta}{1+\beta}  - s\frac{M}{N} - \left(1 - \frac{K\beta}{s}\right)^+\right]}{\frac{\beta N}{s} + 1} \geq \frac{N\left[\frac{2\times 0.991}{1+ 0.991}  - 0.4983 - \left(1 - \frac{0.991}{0.4983}\frac{KM}{N}\right)^+\right]}{\frac{0.991 N}{0.4983 \frac{N}{M} -1} + 1}\nonumber\\
					&\geq \frac{N\left[\frac{2\times 0.991}{1+ 0.991}  - 0.4983 - \left(1 - \frac{0.991\times 1.01}{0.4983}\right)^+\right]}{\frac{0.991 N}{0.4983 \frac{N}{M} -1} + 1}\nonumber\\
					&\geq \frac{\left(0.4983 \frac{N}{M} -1\right)\left[\frac{2\times 0.991}{1+ 0.991}  - 0.4983 - \left(1 - \frac{0.991\times 1.01}{0.4983}\right)^+\right]}{0.991  + 0.4983 \frac{1}{M} - \frac{1}{N}}\nonumber\\
					&\geq \frac{N}{M}\frac{\left(0.4983 - \frac{M}{N}\right)\left[\frac{2\times 0.991}{1+ 0.991}  - 0.4983 - \left(1 - \frac{0.991\times 1.01}{0.4983}\right)^+\right]}{0.991  + \frac{0.4983}{1.01}}\nonumber\\
					&\geq \frac{N}{M}\frac{\left(0.4983 - 0.1250\right)}{0.991  + \frac{0.4983}{1.01}}  \Bigg[\frac{2\times 0.991}{1+ 0.991}  - 0.4983 - \left(1 - \frac{0.991\times 1.01}{0.4983}\right)^+\Bigg] \geq \frac{N}{8M}.
\end{align}

Again, from \cite[Theorem 1]{Maddah-Ali}, we have:
\begin{align}
R_{\mathsf{cen}}(M) &\leq \frac{\min\{N,K\}}{1 + \frac{KM}{N}}\left(1 - \frac{M}{N} \right) \leq \frac{K}{\frac{KM}{N}}\left(1 - \frac{M}{N} \right) \nonumber\\
				  &\leq \frac{N}{M}\left(1 - \frac{M}{N} \right)\leq \frac{N}{M}.
\end{align}
Thus for this regime, the gap between the upper and lower bound is given by:
\begin{align}\label{eq:gapc2r2}
\mathrm{Gap} = \frac{R_{\mathsf{cen}}(M)}{R_{\mathsf{cen}}^*(M)} \leq 8.
\end{align}

\item \textbf{\textit{Regime}} $\bm{\mathit{3 \left(0.1250 N < M \leq N \right):}}$

\n In this regime, setting $s = 1$ and $\beta = 1$ i.e., $\ell = N$, we have:
\begin{align}
R_{\mathsf{cen}}^*(M) \geq \frac{N - M}{N} = \left( 1 - \frac{M}{N}\right).
\end{align}
Again, from \cite[Theorem 1]{Maddah-Ali}, we have:
\begin{align}
R_{\mathsf{cen}}(M) &\leq \frac{\min\{N,K\}}{1 + \frac{KM}{N}}\left(1 - \frac{M}{N} \right) \leq \frac{K}{\frac{KM}{N}}\left(1 - \frac{M}{N} \right) \nonumber\\
					&\leq \frac{N}{M}\left(1 - \frac{M}{N} \right) \leq \frac{1}{0.1250}\left(1 - \frac{M}{N} \right).
\end{align}
Thus for this regime, the gap between the upper and lower bound is given by:
\begin{align}\label{eq:gapc2r3}
\mathrm{Gap} = \frac{R_{\mathsf{cen}}(M)}{R_{\mathsf{cen}}^*(M)}\leq \frac{1}{0.1250} = 8.
\end{align}
\end{itemize}

Thus from (\ref{eq:gapc1}), (\ref{eq:gapc2r1}), (\ref{eq:gapc2r2}) and (\ref{eq:gapc2r3}), we have for all $N,K$, the gap between the achievability and the proposed converse is upper bounded by $8$. This completes proof of Theorem \ref{th:gapfund}. \qed


\section{Proof of Theorem \ref{th:d2d_ldem}}\label{ap:d2d_ldem}

Consider the case of D2D-assisted content delivery for cache-aided system with a library of $N\in\mathbb{N}^+$ files $F_{[1:N]}$ each of size $B$ bits, and $K\in\mathbb{N}^+$ users, with cache storage $Z_{[1:K]}$ which satisfies the minimum D2D storage constraint $KM\geq N$. Let $s$ be an integer such that $s \in \left[1:\min\{\lceil N/L\rceil,K\}\right]$. The demand vector is such that each user requests $L$ distinct files at each transmission interval. Consider the first $s$ caches $Z_{[1:s]}$ and a demand vector 
\begin{align}
\mathbf{D}_1 &= \left(\mathbf{d}_{[1:s]},\mathbf{d}_{[s+1:K]}\right)=  \Bigg(\underbrace{\{1:L\},\{L+1:2L\},\ldots,\{L(s-1)+1:Ls\}}_{=~\mathbf{d}_{[1:s]}},\phi\Bigg),
\end{align}
where the first $s$ user demands are for $Ls$ unique files and last $K-s$ users' demands can be for any arbitrary $L(K-s)$ files. To service this set of demands, consider a composite multicast transmission
\begin{align}
X_1 = \left\{ 
											X^{s+1}_{ \left(\mathbf{d}_{[1:s]},\phi\right)},
											X^{s+2}_{\left(\mathbf{d}_{[1:s]},\phi\right)},
											\ldots,
											X^{K}_{\left(\mathbf{d}_{[1:s]},\phi\right)}
											\right\},
\end{align}
composed of $(K-s)$ multicast transmissions, which, along with the $s$ device caches decodes the files $F_{[1:Ls]}$. Similarly consider another demand vector, 
\begin{align}
\mathbf{D}_2 = \Big(\{Ls+1:L(s+1)\},\{L(s+1)+1:L(s+2)\},\ldots,\{L(2s-1):2Ls\},\phi\Big).
\end{align}
A second composite multicast transmission $X_2$,
along with device cache contents $Z_{[1:s]}$, can decode the next $Ls$ files $F_{[{Ls+1}:{2Ls}]}$. Thus considering the request vectors $\mathbf{D}_1,\mathbf{D}_2,\ldots,\mathbf{D}_{\lceil N/(Ls)\rceil}$ and their corresponding composite multicast transmissions $X_1,X_2,\ldots,X_{\lceil N/(Ls) \rceil}$, along with the first $s$ device caches $Z_{[1:s]}$, the whole library of files $F_{[1:N]}$ can be decoded. 
 Note that for an optimal composite transmission rate $R_{\mathsf{d2d}}^*(M)$, each device in the D2D cluster multicasts with a rate of $R_{\mathsf{d2d}}^*(M)/K$. Thus the entropy of each composite transmission, consisting of $K-s$ transmissions, can be upper bounded as
\begin{align}\label{eq:ldem_d2d_rate}
&H\left(X_i\right)~ \leq~ \frac{(K-s)}{K}R_{\mathsf{d2d,L}}^*(M), ~~~~~~\forall i \in \left[1: \lceil N/(Ls) \rceil\right]
\end{align}
Considering $B=1$ without loss of generality, we have the following sequence of inequalities
\begin{align}
N &\leq I\left(F_{1:N};Z_{[1:s]}, {X}_{[1:{\lceil N/(Ls) \rceil}]}\right)\leq H\left( Z_{[1:s]}, {X}_{[1:{\lceil N/(Ls) \rceil}]}  \right) \leq H\left( Z_{[1:s]} \right) + H\left( {X}_{[1:{\lceil{N/(Ls)}\rceil}]} | Z_{[1:s]}\right)\nonumber\\
	& \leq sM + H\left(  {X}_{[1:{\lceil{N/(Ls)}\rceil}]} | Z_{[1:s]} \right) \leq sM + H\left( {X}_{[1:{\ell}]}| Z_{[1:s]}\right) + H\left( {X}_{[\ell + 1 : {\lceil{N/(Ls)}\rceil}]}| Z_{[1:s]}, {X}_{[1:{\ell}]} \right)\nonumber\\
	& \myleq{(a)} sM + \frac{\ell(K-s)}{K} R_{\mathsf{d2d,L}}^*(M) +  H\left({X}_{\left[\ell + 1 : {\lceil{N/(Ls)}\rceil}\right]}| Z_{[1:s]}, {X}_{[1:{\ell}]}, F_{[1:{L\ell s}]} \right)\nonumber\\
	& \myleq{(b)} sM + \frac{\ell(K-s)}{K} R_{\mathsf{d2d,L}}^*(M) + H\big( {X}_{[\ell + 1: \lceil{N/(Ls)}\rceil]},Z_{[s+1:s+\mu]}| Z_{[1:s]}, {X}_{[1:{\ell}]}, F_{[1:{L\ell s}]} \big)\nonumber\\
	& \leq sM + \frac{\ell(K-s)}{K} R_{\mathsf{d2d,L}}^*(M) + \underbrace{H\left(  Z_{[{s+1}:{s+\mu}]}| Z_{[1:s]}, {X}_{[1:{\ell}]}, F_{[1:{L \ell s}]} \right)}_{\triangleq ~\delta} \nonumber\\
	& \hspace{195pt}+ \underbrace{H\left( {X}_{[{\ell + 1}:{\lceil{N/(Ls)}\rceil}]}|  Z_{[1:{s+\mu}]}, {X}_{[1:{\ell}]}, F_{[1:{L\ell s}]} \right)}_{\triangleq~\lambda}\label{Ldth1:7},
\end{align}
where step \textsf{(a)} follows from \eqref{eq:ldem_d2d_rate} and that the device storage contents, $Z_{[1:s]}$, along with the composite transmission vectors $X_{[1:{\ell}]}$ are capable of decoding the files $F_{[1:{L\ell s}]}$. In step \textsf{(b)}, $\mu = \left(\min\left\{ \left\lceil N/(L\ell)\right\rceil, K\right\} - s\right)$ is the number of additional device caches which, along with the transmissions $X_{[1:{\ell}]}$ can decode all $N$ files. Note that, for $s=K$, we have:
\begin{align}
H\left( X_{[1:{\lceil{N/(Ls)}\rceil}]} | Z_{[1:s]} \right) = 0,
\end{align}
since transmissions are functions of all $K$ caches. As a result, the second step in \eqref{Ldth1:7} yields the minimum storage constraint for D2D-assisted delivery $KM\geq N$. Next we upper bound the terms $\delta,\lambda$ in \eqref{Ldth1:7} which finally yields an upper bound on the RHS. We first note that the term $\delta$ is identical to the case of centralized delivery and can be upper bounded using Han's Inequality by following the same steps as in \eqref{Lth1d:2}-\eqref{eq:Lth1d:7} in Appendix \ref{ap:ldem}, yielding the upper bound
\begin{align}\label{eq:del_d2d}
\delta \leq \frac{\mu}{s + \mu}(N - L\ell s)^{+}. 
\end{align}


\n \textbf{\textit{Upper Bound on}} $\bm{\lambda}:$ We next derive an upper bound on the factor $\lambda$ in \eqref{Ldth1:7} and consider two distinct cases as follows.

\n $\bullet$ \textbf{\textit{Case}} $\bm{\mathit{1 \left(N \leq L\ell(s + \mu)\right):}}$ We consider the case that all $N$ files can be decoded with $\mu \leq K-s$ additional device storage contents and transmissions $X_{[1:{\ell}]}$, within the conditioning in the factor $\lambda$ in (\ref{Ldth1:7}), i.e., $L\ell(s + \mu )\geq N$. Thus, we have
\begin{align}
\lambda &= H\left( X_{[{\ell + 1}:{\lceil{N/(Ls)}\rceil}]}|  Z_{[1:{s+\mu}]}, F_{[1:{N}]} \right) =0, \label{th2l:1}
\end{align}
which follows from the fact that the transmissions are functions of all $N$ files.

\n $\bullet$ \textbf{\textit{Case}} $\bm{\mathit{2 \left(N > L\ell(s + \mu)\right):}}$ We consider the complementary case where $\mu = K - s$ additional device storage contents along with the transmissions $X_{[1:{\ell}]}$, cannot decode all $N$ files. We have:
\begin{align}
\lambda &= H\left( X_{[{\ell + 1}:{\lceil{N/(Ls)}\rceil}]}|  Z_{[1:{K}]}, F_{[1:{KL\ell}]} \right) \leq  H\left( X_{[{\ell + 1}:{\lceil{N/(Ls)}\rceil}]}|  Z_{[1:{K}]}\right) = 0, \label{th2l:2}
\end{align}
which follows from the fact that $KM\geq N$ i.e., all files are stored within the collective device caches for D2D-assisted delivery and hence all transmissions are functions of the cache contents. Thus combining (\ref{th2l:1}) and (\ref{th2l:2}) we have: 
\begin{align}\label{Ldth2l:3}
\lambda = 0.
\end{align}
Substituting (\ref{eq:del_d2d}) and (\ref{Ldth2l:3}) into (\ref{Ldth1:7}) 
and optimizing over all parameter values of $s,\ell$, 
 completes the proof of Theorem \ref{th:d2d_ldem}. \qed

\section{Proof of Theorem \ref{th:d2d_ldem_gap}}\label{ap:d2d_ldem_gap}

\newcommand{\al}{\alpha}
\renewcommand{\mo}{\mu_1} 
\renewcommand{\lo}{\lambda_1} 
\renewcommand{\lt}{\lambda_2} 
\renewcommand{\bo}{\beta_1} 
\renewcommand{\bt}{\beta_2} 
\renewcommand{\con}{10} 

\newcommand{\ceiling}[1]{\left\lceil #1 \right\rceil}
\newcommand{\flooring}[1]{\left\lfloor #1 \right\rfloor}

\renewcommand{\al}{0.5}
\renewcommand{\mo}{0.51} 
\renewcommand{\lo}{0.2} 
\renewcommand{\lt}{\frac{1}{3}} 
\renewcommand{\bo}{0.984} 
\renewcommand{\bt}{0.5} 
\renewcommand{\con}{10}

From Theorem \ref{th:d2d_ldem}, considering the lower bound on the optimal rate $R_{\mathsf{d2d,L}}^*(M)$, we set $\ell = \left\lceil \frac{\beta N}{Ls}\right\rceil \in \left[1:\left\lceil \frac{N}{Ls}\right\rceil \right]$ with $\beta \in [0,1]$. We make use of the upper bound on $\left(\frac{\mu}{\mu + s} \right)$ from \eqref{eq:boundmu_ldem} in Appendix \ref{ap:ldem_gap}. Using this in \eqref{eq:d2d_ldem} from Theorem \ref{th:d2d_ldem}, we have
\begin{align}\label{eq:Ld2d_lb}
R_{\mathsf{d2d,L}}^*(M) & \geq \frac{N - sM - \frac{1}{1 + \beta}\left(N - L\left\lceil\frac{\beta N}{Ls}\right\rceil s\right)^+}{\left\lceil\frac{\beta N}{Ls}\right\rceil\left(\frac{K-s}{K}\right)} \geq \frac{N\left(\frac{2\beta}{1+\beta} - s\frac{M}{N}\right)}{\left\lceil\frac{\beta N}{Ls}\right\rceil\left(\frac{K-s}{K}\right)}
\end{align}

In order to facilitate the proof of Theorem \ref{th:d2d_ldem_gap}, we consider two cases namely - $(i)$ \textit{low per-device demand} with $\al N \geq L$; and $(ii)$ \textit{high per-device demand} with $\al N \leq L$. We consider the two cases separately.

\n $\bullet$ \textbf{\textit{Case}} $\bm{\mathit{1~ \left(\al N \geq L \right) }}$: For the case of low-per device demands, we divide the available cache storage at each device into the following three regimes, namely $(i)$ \textit{Regime 1}: $N/K \leq M \leq L$; $(ii)$ \textit{Regime 2}: $L \leq M \leq \lo N$; and $(iii)$ \textit{Regime 3}: $\lo N \leq M \leq N$. We consider each regime separately.

\begin{itemize}
\item \textbf{\textit{Regime}} $\bm{\mathit{1~ \left(N/K \leq M \leq L\right):}}$ 

\n For this regime of cache storage, we consider two further sub-cases, i.e., $(i)$ $N<K$ and $(ii)$ $N\geq K$. We next treat each of the sub-cases separately.

\n $-$ \textbf{\textit{Sub-case}} $\bm{\mathit{1~ \left(N<K\right)}}$: For this sub-case, we note that from the minimum storage constraint for D2D-assisted delivery, i.e., $KM\geq N$, the minimum allowable cache storage at each user can be less than unity. Therefore, we divide the available cache storage in this regime into two sub-regimes namely $(i)$ $N/K\leq M \leq \al$	 and $(ii)$ $\al \leq M\leq L$. We these sub-regimes separately as follows. Consider first, the sub-regime i.e., $N/K\leq M \leq \al$. For this sub-regime consider the case when $N=1$. For this case, setting $s=1$ and $\beta=1$, from the lower bound in \eqref{eq:Ld2d_lb}, we have 
\begin{align}
R^*_{\mathsf{d2d,L}}(M) \geq (1 - M),
\end{align}
where we have used the fact that $L=1$ when $N=1$. Again considering the upper bound in \eqref{eq:d2d_ldem_ub}, we have $R_{\mathsf{d2d,L}}\leq 1$. Using the upper and the lower bounds, we have
\begin{align}\label{eq:Ld2d_gap1}
\mathrm{Gap}  = \frac{R_{\mathsf{d2d,L}}}{R^*_{\mathsf{d2d,L}}} \leq \frac{1}{1-M} \leq \frac{1}{1 - \al} = 2.
\end{align}
Next, we consider the case when $N\geq 2$. For this case, setting $s = \lceil \frac{N}{L} \rceil \in \left[1:\lceil \frac{N}{L} \rceil\right]$ and $\beta=1$, from \eqref{eq:Ld2d_lb}, we have
\begin{align}
R^*_{\mathsf{d2d,L}}(M) &\geq \frac{N\left(1 - \ceiling{\frac{N}{L}} \frac{M}{N} \right)}{\ceiling{\frac{N}{L\ceiling{\frac{N}{L}}}} \frac{K-\ceiling{\frac{N}{L}}}{K}} \geq N\left(1 - \left(\frac{N}{L}+1\right)\frac{M}{N} \right)\nonumber\\
											  &= N\left(1 - \left(\frac{1}{L}+\frac{1}{N}\right)M \right) \mygeq{(a)} N\left( 1 - \frac{3}{2}\times\al \right),
\end{align}
where step \textsf{(a)} follows from the fact that $N\geq 2$ and $L\geq 1$. Again, from the upper bound in \eqref{eq:d2d_ldem_ub}, we have $R_{\mathsf{d2d,L}}\leq N$. Using this, we have
\begin{align}\label{eq:Ld2d_gap2}
\mathrm{Gap} = \frac{R_{\mathsf{d2d,L}}}{R^*_{\mathsf{d2d,L}}} \leq \frac{1}{1 - \frac{3}{2}\times\al} = 4.
\end{align}

We next consider the sub-regime $\al \leq M \leq L$.
In this regime, setting $s = \flooring{\al \frac{N}{M}} \in \left[1: K\right]$\footnote{The regime of $s$ can be verified as follows. Using the lower bound $\al \leq M$, we have $\al N/M \leq N < K$. Again using the upper bound $M\leq L$, we have $\al N/M \geq \al N/L \geq 1$.} and $\beta=1$, from the lower bound in \eqref{eq:Ld2d_lb}, we have
\begin{align}
R^*_{\mathsf{d2d,L}} &\geq \frac{N\left(1 -  \flooring{\al\frac{N}{M}}\frac{M}{N}\right)}{\ceiling{\frac{N}{L\flooring{\al \frac{N}{M}}}}\frac{K - \flooring{\al\frac{N}{M}}}{K}}\geq \frac{N\left(1 - \al \right)}{\ceiling{\frac{N/L}{\flooring{\al\frac{N}{L}}}}} \mygeq{(a)} \frac{N(1 - \al)}{3},
\end{align}
where step \textsf{(a)} follows from the fact that for any $N/L \geq 2$, we have $\frac{N/L}{\flooring{\al(N/L)}}\leq 3$. Again from the upper bound in \eqref{eq:d2d_ldem_ub}, we have $R_{\mathsf{d2d,L}}(M)\leq N$. Using the upper and lower bounds, we have
\begin{align}\label{eq:Ld2d_gap3}
\mathrm{Gap} = \frac{R_{\mathsf{d2d,L}}(M)}{R_{\mathsf{d2d,L}}^*(M)} \leq \frac{3}{1 - \al} = 6.
\end{align}

\n $-$ \textbf{\textit{Sub-case}} $\bm{\mathit{2~ \left(N\geq K\right)}}$: For this sub-case, we note that from the minimum storage constraint for D2D-assisted delivery, i.e., $KM\geq N$, we have $M\geq 1$. Therefore, we consider the following regime of available cache storage $\al \leq N/K \leq M \leq L$. In this regime, setting $s = \flooring{\al \frac{N}{M}} \in \left[1: K\right]$\footnote{The regime of $s$ is validated as follows. Using the lower bound $M\geq N/K$, we have $\al N/M \leq \al K \leq K$. Again using the upper bound $M\leq L$, we have $\al N/M \geq \al N/L \geq 1$.} and $\beta=1$, from the lower bound in \eqref{eq:Ld2d_lb}, we have
\begin{align}
R^*_{\mathsf{d2d,L}} &\geq \frac{N\left(1 -  \flooring{\al\frac{N}{M}}\frac{M}{N}\right)}{\ceiling{\frac{N}{L\flooring{\al \frac{N}{M}}}}\frac{K - \flooring{\al\frac{N}{M}}}{K}}\geq \frac{N\left(1 - \al \right)}{\ceiling{\frac{N/L}{\flooring{\al\frac{N}{L}}}}} \mygeq{(a)} \frac{N(1 - \al)}{3},
\end{align}
where step \textsf{(a)} again follows from the fact that for any $N/L \geq 2$, we have $\frac{N/L}{\flooring{\al(N/L)}}\leq 3$. Again from the upper bound in \eqref{eq:d2d_ldem_ub}, we have $R_{\mathsf{d2d,L}}(M)\leq N$. Using the upper and lower bounds, we have
\begin{align}\label{eq:Ld2d_gap4}
\mathrm{Gap} = \frac{R_{\mathsf{d2d,L}}(M)}{R_{\mathsf{d2d,L}}^*(M)} \leq \frac{3}{1 - \al} = 6.
\end{align}

\item \textbf{\textit{Regime}} $\bm{\mathit{2~ \left(L \leq M \leq \lo N\right):}}$

\n For this regime, setting $s = \left\lfloor \mo \frac{N}{M}\right\rfloor \in [1:K]$\footnote{The regime of $s$ can be validated as follows. Consider first, a lower bound on $\al N/M$. In the given regime, we have $ \al N/M \geq \al/\lo \geq 1$. Next, we consider an upper bound on $\al N/M$. Consider first, the case when $N/L \leq K$. In this case, its easy to note that $\al N/M \leq K$. Next consider the case that $N/L\geq K$. In this case, \textit{Regime} $\mathit{2}$ reduces to $L\leq N/K \leq M \leq \lo N$ due to the minimum storage constraint and hence we have $\al N/M \leq \al K \leq K$. Therefore we have $\flooring{\al N/M} \in [1:K]$.} and $\ell = \left\lceil \frac{\bo N}{Ls}\right\rceil$, from \eqref{eq:Ld2d_lb}, we have
\begin{align}\label{eq:reg2}
&\hspace{-15pt}R^*_{\mathsf{d2d,L}}(M) ~\geq \frac{\left(\frac{2\times\bo}{1 + \bo}\right) - s\frac{M}{N}}{\left(\frac{\bo}{Ls} + \frac{1}{N}\right)\left[\frac{K-s}{K}\right]} \myeq{(a)} \frac{\left(\frac{2\times\bo}{1 + \bo}\right) - \left\lfloor \mo \frac{N}{M}\right\rfloor\frac{M}{N} }{\frac{\bo}{L\left\lfloor \mo \frac{N}{M}\right\rfloor} + \frac{1}{N}} \mygeq{(b)} \frac{\left(\frac{2\times\bo}{1 + \bo}\right) -  \mo \frac{N}{M}\frac{M}{N}}{\frac{\bo}{L\left(\mo \frac{N}{M} - 1\right)} + \frac{1}{N}}\nonumber\\
												&\hspace{-15pt} \geq \frac{\frac{LN}{M} \left(\mo - \frac{M}{N} \right)\left\{\left(\frac{2\times\bo}{1 + \bo}\right) -  \mo \right\}}{\bo + \mo\left(\frac{L}{M} - \frac{L}{N}\right)}  ~\mygeq{(c)} \frac{\frac{LN}{M} \left(\mo - \lo \right)\left\{\left(\frac{2\times\bo}{1 + \bo}\right) -  \mo \right\}}{\bo + \mo} \geq \frac{LN}{10M},
\end{align}
where step \textsf{(a)} follows due to the fact that $(K-s)/K \leq 1$; step \textsf{(b)} follows by using $\lfloor \mo N/M \rfloor \leq \mo N/M $ in the numerator and $\lfloor \mo N/M\rfloor \leq \mo N/M-1$ in the denominator; and step \textsf{(c)} follows by using $M/N\leq \lo$ in the numerator and $M\geq L$ in the denominator. Again, considering the upper bound in \eqref{eq:d2d_ldem_ub}, we have
\begin{align}
R_{\mathsf{d2d,L}}(M)\leq  \frac{LN}{M}\left(1 - \frac{M}{N}\right)\leq \frac{LN}{M}.
\end{align}
Therefore for \textit{Regime $\mathit{2}$}, we have
\begin{align}\label{eq:Ld2d_gap5}
\mathrm{Gap} = \frac{R_{\mathsf{d2d,L}}(M)}{R^*_{\mathsf{d2d,L}}(M)} \leq 10.
\end{align}

\item \textbf{\textit{Regime}} $\bm{\mathit{3~ \left(\lo N\leq M \leq N\right):}}$

\n In this regime, setting $s=1$ and $\beta = 1$ in \eqref{eq:Ld2d_lb}, we have
\begin{align}
R^*_{\mathsf{d2d,L}}(M) & \geq \frac{N\left(1 - \frac{M}{N}\right)}{\frac{N}{L}+1} \geq \frac{\left(1 - \frac{M}{N}\right)}{\frac{1}{L}+\frac{1}{N}}.
\end{align}
Again, considering the upper bound in \eqref{eq:d2d_ldem_ub}, we have
\begin{align}
R_{\mathsf{d2d,L}}(M)\leq \frac{LN}{M}\left(1 - \frac{M}{N}\right).
\end{align}
Therefore for \textit{Regime 3}, we have
\begin{align}\label{eq:Ld2d_gap6}
\mathrm{Gap} = \frac{R_{\mathsf{d2d,L}}(M)}{R^*_{\mathsf{d2d,L}}(M)} \leq \frac{LN}{M}\left(\frac{1}{L}+\frac{1}{N}\right) \myleq{(a)}  \frac{LN}{M} \times \frac{2}{L} \leq \frac{2N}{M} \leq \frac{2}{\lo} \leq 10,
\end{align}
where step \textsf{(a)} follows from the fact that $L\leq N$.
\end{itemize}

\renewcommand{\con}{C} 

\n $\bullet$ \textbf{\textit{Case}} $\bm{\mathit{2~ \left(\al N \leq L \right) }}$:
For the case of high per-device demands,  we divide the available cache storage at each device into the following two regimes, namely $(i)$ \textit{Regime} $\mathit{1}$: $N/K \leq M \leq N/3$; and $(ii)$ \textit{Regime} $\mathit{2}$: $N/3 \leq M \leq N$. We next consider each regime separately.

\begin{itemize}
\item \textbf{\textit{Regime}} $\bm{\mathit{1 \left(N/K \leq M \leq N/3\right):}}$

\n For this regime, setting $s = 1$ and $\ell = \left\lceil \frac{\bt N}{Ls}\right\rceil$, from \eqref{eq:Ld2d_lb}, we have
\begin{align}
R^*_{\mathsf{d2d,L}}(M) &\geq \frac{N\left(\left(\frac{2\times\bt}{1 + \bt}\right) - \frac{M}{N}\right)}{\ceiling{\frac{\bt N}{L}}\left(\frac{K-1}{K}\right)} 				\geq 							\frac{N\left(\left(\frac{2\times\bt}{1 + \bt}\right) - \frac{M}{N}\right)}{\frac{\bt N}{L} + 1} \mygeq{(a)} \frac{N\left(\left(\frac{2\times\bt}{1 + \bt}\right) - \lt \right)}{2}, 
\end{align}
where step \textsf{(a)} follows by using the lower bound $L\geq 0.5N$. Again, considering the upper bound in \eqref{eq:d2d_ldem_ub}, we have $R_{\mathsf{d2d,L}}(M)\leq N$. Using the upper and lower bounds, we have
\begin{align}\label{eq:Ld2d_gap7}
\mathrm{Gap} =  \frac{R_{\mathsf{d2d,L}}(M)}{R^*_{\mathsf{d2d,L}}(M)} \leq  \frac{2}{\left(\frac{2\times\bt}{1 + \bt}\right) - \lt} \leq 6.
\end{align}

\item \textbf{\textit{Regime}} $\bm{\mathit{2 \left( N/3 \leq M \leq N\right):}}$

\n In this regime, setting $s=1$ and $\beta = 1$ in \eqref{eq:Ld2d_lb}, we have
\begin{align}
R^*_{\mathsf{d2d,L}}(M) & \geq \frac{N\left(1 - \frac{M}{N}\right)}{\frac{N}{L}+1} \geq \frac{\left(1 - \frac{M}{N}\right)}{\frac{1}{L}+\frac{1}{N}}.
\end{align}
From \eqref{eq:d2d_ldem_ub}, we have
\begin{align}
R_{\mathsf{d2d,L}}(M)\leq \frac{LN}{M}\left(1 - \frac{M}{N}\right).
\end{align}
Therefore for \textit{Regime 3}, we have
\begin{align}\label{eq:Ld2d_gap8}
\mathrm{Gap} = \frac{R_{\mathsf{d2d,L}}(M)}{R^*_{\mathsf{d2d,L}}(M)} \leq \frac{LN}{M}\left(\frac{1}{L}+\frac{1}{N}\right) \leq  \frac{LN}{M} \times \frac{2}{L} \leq \frac{2N}{M} \leq \frac{2}{1/3} = 6
\end{align}
\end{itemize}
Finally, combining \eqref{eq:Ld2d_gap1}, \eqref{eq:Ld2d_gap2}, \eqref{eq:Ld2d_gap3}, \eqref{eq:Ld2d_gap4}, \eqref{eq:Ld2d_gap5}, \eqref{eq:Ld2d_gap6}, \eqref{eq:Ld2d_gap7} and \eqref{eq:Ld2d_gap8}, completes the proof of Theorem \ref{th:d2d_ldem_gap}. \qed

\section{Proof of Theorem \ref{th:gapd2d}}\label{ap:gapd2d}

From Corollary \ref{cor:convd2d}, considering the lower bound on the optimal rate $R_{\mathsf{d2d}}^*(M)$, we set $\ell = \left\lceil \frac{\beta N}{s}\right\rceil \in \{1,2,\ldots,\left\lceil \frac{N}{s}\right\rceil \}$ with $0 <\beta \leq 1$. Under this setting, we can again use the upper bound on $\left(\frac{\mu}{\mu + s} \right)$ from \eqref{eq:boundmu} from Appendix \ref{ap:gapfund}. Substituting \eqref{eq:boundmu} into (\ref{eq:convd2d}), we have the following form on the lower bound
\begin{align}\label{eq:convlbd2d}
R_{\mathsf{d2d}}^*(M) &\geq \frac{N - sM - \frac{1}{1 + \beta}\left(N - \left\lceil\frac{\beta N}{s}\right\rceil s\right)^+ }{\left\lceil\frac{\beta N}{s}\right\rceil \left(\frac{K-s}{K}\right)} \geq \frac{N\left(\frac{2\beta}{1 + \beta} - s\frac{M}{N} \right)}{\left\lceil\frac{\beta N}{s}\right\rceil \left(\frac{K-s}{K}\right)}.
\end{align}

\n In order to facilitate the proof of Theorem \ref{th:gapd2d}, we consider the three cases, namely $(i)$ the gap at $M=N/K$; $(ii)$ the case when $N<K$; and $(iii)$ the case when $N\geq K$. We next consider each of these cases separately.

\n $\bullet$ \textbf{\textit{Case}} $\bm{\mathit{1}}$ $\left(\textbf{\textit{Gap at}}~ \bm{\mathit{M = N/K}}\right)$: In this case, setting $s = N$ and $\beta  = 1$ in (\ref{eq:convlbd2d}), we have
\begin{align}
R_{\mathsf{d2d}}^*(M) &\geq \frac{N\left(1 - M\right)}{\left(\frac{K-N}{K}\right)} = \frac{KN\left(1 - \frac{N}{K}\right)}{K-N} = N.
\end{align}
Again, from \cite[Theorem 1]{fund_ji}, we have
\begin{align}
R_{\mathsf{d2d}}(M) \leq  \frac{N}{M}\left(1 - \frac{M}{N}\right)\leq N,
\end{align}
where the last inequality stems from the fact that for $N\geq K$, $M\geq 1$. Thus the gap is given by
\begin{align}\label{eq:convd2dgap1}
\mathrm{Gap} = \frac{R_{\mathsf{d2d}}(M)}{R_{\mathsf{d2d}}^*(M)} \leq  1.
\end{align}

\n 
\n $\bullet$ \textbf{\textit{Case}} $\bm{\mathit{2 \left(N < K \right)}}:$ For $N < K$ we divide the cache memory into $3$ distinct regimes: \textit{Regime 1:} ~ $N/K < M \leq 1$; {\textit{Regime 2:}}~$1 < M \leq 0.1250N$; and {\textit{Regime 3:}} ~ $0.1250N < M \leq N$. We consider each of the regimes separately.

\begin{itemize}
\item \textbf{\textit{Regime}}~$\bm{\mathit{1 \left(N/K < M \leq 1\right):}}$ 

\n For this regime, we consider two sub-regimes, namely $(i)$ $N/K<M\leq 2/3$; and $(ii)$ $2/3 < M \leq 1$. We consider each sub-regime separately. 

\n $-$ \textbf{\textit{Sub-regime}} $\bm{\mathbf{1 \left(N/K < M \leq 2/3\right):}}$  From (\ref{eq:convlbd2d}), setting $s = N$ and $\beta  = 1$, we have
\begin{align}
R_{\mathsf{d2d}}^*(M) &\geq \frac{N\left(1 - M\right)}{\left(\frac{K-N}{K}\right)} \geq N(1-M) \geq N\left(1 - \frac{2}{3}\right) = \frac{N}{3}
\end{align}
Again from \cite[Theorem 1]{fund_ji}, we have
\begin{align}
R_{\mathsf{d2d}} \triangleq \min\left\{\frac{N}{M}\left(1 - \frac{M}{N}\right), N \right\} \leq N
\end{align}
Thus in this sub-regime, the gap between the upper and lower bounds is given by
\begin{align}\label{eq:convd2dgap2}
\mathrm{Gap} = \frac{R_{\mathsf{d2d}}(M)}{R_{\mathsf{d2d}}^*(M)} \leq \frac{N}{N/3} \leq 3.
\end{align}

\n $-$ \textbf{\textit{Sub-regime}} $\bm{\mathbf{2 \left(2/3 < M \leq 1\right):}}$ Consider first, the case when $N=1$. Setting $s=1, \beta = 1$ in (\ref{eq:convlbd2d}), we have
\begin{align}
R_{\mathsf{d2d}}^*(M) \geq \frac{1 - sM}{\lceil 1/s \rceil} \geq 1 - M.
\end{align}
From \cite[Theorem 1]{fund_ji}, setting $N=1$, we have
\begin{align}
R_{\mathsf{d2d}}(M) \leq \frac{1}{M}(1- M).
\end{align}
Thus the gap in this case is given by
\begin{align}\label{eq:convd2dgap3a}
\mathrm{Gap} = \frac{R_{\mathsf{d2d}}(M)}{R_{\mathsf{d2d}}^*(M)} \leq \frac{1}{M} \leq \frac{3}{2}.
\end{align}

\n For the case of $N\geq 2$, setting $s = \left\lfloor \frac{2N}{3M}\right\rfloor$ and $\beta =1$ in (\ref{eq:convlbd2d}), we have
\begin{align}
R_{\mathsf{d2d}}^*(M) &\geq \frac{N\left(1 - \left\lfloor \frac{2N}{3M}\right\rfloor\frac{M}{N} \right)}{\left\lceil\frac{N}{\left\lfloor \frac{2N}{3M}\right\rfloor}\right\rceil \left(\frac{K-s}{K}\right)}\geq \frac{N\left(1 - \left\lfloor \frac{2N}{3M}\right\rfloor\frac{M}{N} \right)}{\left\lceil\frac{N}{\left\lfloor \frac{2}{3} N\right\rfloor}\right\rceil } \geq \frac{N\left(1 -  \frac{2}{3}\right)}{\left\lceil 2\right\rceil} \geq \frac{N}{6}, 
\end{align}
where, we have used the fact that $\frac{N}{\left\lfloor 2N/3\right\rfloor} \leq 2, ~~\forall N\geq 2$. Again from \cite[Theorem 1]{fund_ji}, we have
\begin{align}
R_{\mathsf{d2d}} \triangleq \min\left\{\frac{N}{M}\left(1 - \frac{M}{N}\right), N \right\} \leq N
\end{align}
Thus in this sub-regime, the gap between the upper and lower bounds is given by
\begin{align}\label{eq:convd2dgap3b}
\mathrm{Gap} = \frac{R_{\mathsf{d2d}}(M)}{R_{\mathsf{d2d}}^*(M)} \leq \frac{N}{N/6} \leq 6.
\end{align}

\n Thus, in general, in this regime, we can upper bound the Gap by the constant $6$.\\

\item \textbf{\textit{Regime}} $\bm{\mathit{2 \left(1 < M \leq 0.1250N\right):}} $

\n For this regime, setting $s = \lfloor 0.5 N/M \rfloor$, $\beta = 1$ in (\ref{eq:convlbd2d}), we have
\begin{align}
R_{\mathsf{d2d}}^*(M) &\geq \frac{N\left(\frac{2\beta}{1+\beta}  - s\frac{M}{N}\right)}{\frac{\beta N}{s} + 1}\times\frac{K}{K - s} \geq \frac{N\left(1  - \lfloor 0.5 N/M \rfloor\frac{M}{N}\right)}{\frac{N}{\lfloor 0.5 N/M \rfloor} + 1}\nonumber\\
					&\geq \frac{N\left(1  - 0.5\right)}{\frac{N}{0.5 N/M - 1} + 1} \geq \frac{0.5N\left(0.5 N/M - 1\right)}{N + \frac{0.5N}{M} - \frac{1}{N}} \geq \frac{0.5\left(0.5 N/M - 1\right)}{1 + \frac{0.5}{M}} \geq \frac{N}{M}\frac{0.5\left(0.5 - \frac{M}{N}\right)}{1 + 0.5}	\nonumber\\				
					&\geq \frac{N}{M}\frac{0.5\left(0.5 - 0.1250\right)}{1 + 0.5} \geq \frac{N}{8M}
\end{align}
Again, from \cite[Theorem 1]{fund_ji}, we have
\begin{align}
R_{\mathsf{d2d}}(M) \leq \frac{N}{M} - 1 \leq \frac{N}{M}.
\end{align}
Thus in this regime, the gap between the upper and lower bounds is given by
\begin{align}
\mathrm{Gap} = \frac{R_{\mathsf{d2d}}(M)}{R_{\mathsf{d2d}}^*(M)} \leq 8.
\end{align}

\item \textbf{\textit{Regime}} $\bm{\mathit{3 \left(0.1250N < M \leq N\right):}}$

\n For this regime, setting $s=1$ and $\beta=1$ in (\ref{eq:convlbd2d}), we have
\begin{align}
R_{\mathsf{d2d}}^*(M)\geq \frac{N - M}{N} = \left(1 - \frac{M}{N}\right)
\end{align}
Again, from \cite[Theorem 1]{fund_ji}, we have
\begin{align}
R_{\mathsf{d2d}}(M) \leq \frac{N}{M}\left(1 - \frac{M}{N}\right) \leq \frac{1}{0.1250}\left(1 - \frac{M}{N}\right).
\end{align}
Thus in this regime, the gap between the upper and lower bounds is given by
\begin{align}
\mathrm{Gap} = \frac{R_{\mathsf{d2d}}(M)}{R_{\mathsf{d2d}}^*(M)} \leq  \frac{1}{0.1250} = 8.
\end{align}
\end{itemize}
Thus, for the case of $K>N$, the gap between the upper and lower bounds is $1$ at $M = N/K$, at most $6$ for $N/K \leq M \leq 1$ and at most $8$ for all $1< M \leq N$.

\n $\bullet$ \textbf{\textit{Case}} $\bm{\mathit{3 \left(N\geq K\right):}}$ For this case, we consider two main sub-cases, namely $(i)$ $ K\leq 8$; and $(ii)$ $K\geq 9$. We treat each case separately.

\n $\bullet$ \textbf{\textit{Sub-case}} ~$\bm{\mathit{1 \left(K\leq 8\right):}}$ For this case, from \cite[Theorem 1]{fund_ji}, we have
\begin{align}\label{eq:minnk1}
R_{\mathsf{d2d}}(M)\leq \frac{N}{M}\left(1 - \frac{M}{N}\right) \leq K\left(1 - \frac{M}{N}\right),
\end{align}
where the last inequality is derived from the minimum storage constraint $KM>N$. Again, from (\ref{eq:convlbd2d}), setting $s=1, \beta = 1$, we have:
\begin{align}
R_{\mathsf{d2d}}^*(M) \geq \frac{N\left(1 - \frac{M}{N}\right)}{\left\lceil N\right\rceil} \geq \left(1 - \frac{M}{N}\right).
\end{align}
Thus the gap between the lower and upper bound is given by:
\begin{align}\label{eq:convd2dgap4}
\mathrm{Gap} = \frac{R_{\mathsf{d2d}}(M)}{R_{\mathsf{d2d}}^*(M)} \leq K \leq 8.
\end{align}

\n $\bullet$ \textbf{\textit{Sub-case}} ~$\bm{\mathit{2 \left(K\geq 8\right):}}$

\n For this case, we divide the cache memory into $3$ distinct regimes namely, $(i)$ Regime $1$: $1 \leq N/K \leq M \leq 1.15 N/K$; $(ii)$ Regime $2$: $ 1.15 N/K < M \leq 0.1250N$; and $(iii)$ Regime $3$: $ 0.1250N < M \leq N $. We next consider each regime separately.
 
\begin{itemize}
\item \textbf{\textit{Regime}} ~$\bm{\mathit{1 \left(1 \leq N/K \leq M \leq 1.15 N/K\right):}}$

\n For this regime, setting $s = \lfloor 0.4361 K \rfloor$ and $\beta = 0.7398$ in (\ref{eq:convlbd2d}) and using the fact that $K \geq 9$, we have
\begin{align}
R_{\mathsf{d2d}}^*(M) &\geq \frac{N\left(\frac{2\beta}{1+\beta}  - s\frac{M}{N}\right)}{\frac{\beta N}{s} + 1}\times\frac{K}{K - s} \geq \frac{N\left(\frac{2\times 0.7398}{1+ 0.7398}  - \lfloor 0.4361 K \rfloor\frac{M}{N}\right)}{\frac{0.7398 N}{\lfloor 0.4361 K \rfloor} + 1}\times\frac{1}{1 - \frac{\lfloor 0.4361 K \rfloor}{K}}\nonumber\\
					&\geq \frac{N\left(\frac{2\times 0.7398}{1+ 0.7398}  - 0.4361\frac{KM}{N}\right)}{\frac{0.7398 N}{0.4361 K -1} + 1}\times\frac{1}{1 - \frac{0.4361 K -1}{K}}\nonumber\\
					&\geq \frac{N \left(0.4361 K -1\right)\left[\frac{2\times 0.7398}{1+ 0.7398}  - 0.4361\frac{KM}{N}\right]}{0.7398 N + 0.4361 K -1}\times\frac{1}{1 + \frac{1}{K} - 0.4361}\nonumber\\
					&\geq K\frac{\left(0.4361 - \frac{1}{9}\right)\left[\frac{2\times 0.7398}{1+ 0.7398}  - 0.4361\frac{KM}{N}\right]}{0.7398 + 0.4361\frac{K}{N} -\frac{1}{N}}\times\frac{1}{1 + \frac{1}{9} - 0.4361}\nonumber\\
					&\geq K\frac{\left(0.4361 - \frac{1}{9}\right)\left[\frac{2\times 0.7398}{1+ 0.7398}  - 0.4361\times 1.15\right]}{0.7398 + 0.4361}\times\frac{1}{1 + \frac{1}{9} - 0.4361} \geq \frac{K}{7}.
\end{align}
Again, from (\ref{eq:minnk1}), we have
\begin{align}
R_{\mathsf{d2d}}(M) \leq K.
\end{align}
Thus in this regime, the gap between the upper and lower bounds is given by
\begin{align}
\mathrm{Gap} = \frac{R_{\mathsf{d2d}}(M)}{R_{\mathsf{d2d}}^*(M)} \leq \frac{K}{K/7} = 7.
\end{align}

\item \textbf{\textit{Regime}} ~$\bm{\mathit{2 \left(1.15 N/K < M \leq 0.1250N \right):}}$

\n For this regime, setting $s = \left\lfloor 0.4470\frac{N}{M} \right\rfloor$ and $\beta = 0.8995$ in (\ref{eq:convlbd2d}), we have
\begin{align}
R_{\mathsf{d2d}}^*(M) &\geq \frac{N\left(\frac{2\beta}{1+\beta}  - s\frac{M}{N}\right)}{\frac{\beta N}{s} + 1}\times\frac{K}{K - s} \geq \frac{N\left(\frac{2\times 0.8995}{1+ 0.8995}  - \left\lfloor 0.4470\frac{N}{M} \right\rfloor \frac{M}{N}\right)}{\frac{0.8995 N}{\left\lfloor 0.4470\frac{N}{M} \right\rfloor} + 1} \geq \frac{N\left(\frac{2\times 0.8995}{1+ 0.8995}  - 0.4470\right)}{\frac{0.8995 N}{0.4470\frac{N}{M} - 1} + 1} \nonumber\\
					&\geq \frac{\left(0.4470\frac{N}{M} - 1\right)\left[\frac{2\times 0.8995}{1+ 0.8995}  - 0.4470\right]}{{0.8995} + \frac{0.4470}{M} - \frac{1}{N}} \geq \frac{N}{M}\frac{\left(0.4470 - \frac{M}{N}\right)\left[\frac{2\times 0.8995}{1+ 0.8995}  - 0.4470\right]}{{0.8995} + \frac{0.4470 K}{1.15N}}\nonumber\\
					&\geq \frac{N}{M}\frac{\left(0.4470 - 0.1250\right)\left[\frac{2\times 0.8995}{1+ 0.8995}  - 0.4470\right]}{{0.8995} + \frac{0.4470}{1.15}}\geq \frac{N}{8M}
\end{align}
Again, from \cite[Theorem 1]{fund_ji}, we have
\begin{align}
R_{\mathsf{d2d}}(M) \leq \frac{N}{M} - 1 \leq \frac{N}{M}.
\end{align}
Thus in this regime, the gap between the upper and lower bounds is given by
\begin{align}
\mathrm{Gap} = \frac{R_{\mathsf{d2d}}(M)}{R_{\mathsf{d2d}}^*(M)} \leq 8.
\end{align}

\item \textbf{\textit{Regime}} ~$\bm{\mathit{3 \left(0.1250N < M \leq N \right):}}$

\n For this regime, setting $s=1$ and $\beta=1$ in (\ref{eq:convlbd2d}), we have
\begin{align}
R_{\mathsf{d2d}}^*(M)\geq \frac{N - M}{N} = \left(1 - \frac{M}{N}\right)
\end{align}
Again, from \cite[Theorem 1]{fund_ji}, we have
\begin{align}
R_{\mathsf{d2d}}(M) \leq \frac{N}{M}\left(1 - \frac{M}{N}\right) \leq \frac{1}{0.1250}\left(1 - \frac{M}{N}\right).
\end{align}
Thus in this regime, the gap between the upper and lower bounds is given by
\begin{align}
\mathrm{Gap} = \frac{R_{\mathsf{d2d}}(M)}{R_{\mathsf{d2d}}^*(M)} \leq  \frac{1}{0.1250} = 8.
\end{align}
\end{itemize}
Thus, for the case of $N\geq K$, the gap between the upper and lower bounds is at most $8$ for all $N/K \leq M \leq N$. Combining the three cases completes the proof of Theorem \ref{th:gapd2d}. \qed


\bibliographystyle{IEEEtran}
\bibliography{references}

\end{document}